\documentclass[journal,11pt,onecolumn,draftclsnofoot]{IEEEtran}
\usepackage[margin=1in]{geometry} 

\usepackage{amsmath,amssymb,amsfonts}
\usepackage{algorithm}
\usepackage{algorithmic}
\usepackage{cite}
\usepackage{graphicx}
\usepackage{textcomp}
\usepackage{xcolor}
\usepackage{subfigure}
\usepackage{bm}
\usepackage{amsthm}
\usepackage{booktabs}
\usepackage{url}
\usepackage[authoryear, square]{natbib} 

\theoremstyle{definition}
\newtheorem{theorem}{\textbf{Theorem}}
\newtheorem{definition}{\textbf{Definition}}

\newtheorem{proposition}{\textbf{Proposition}}
\newtheorem{remark}{\textbf{Remark}}

\newtheorem{lemma}{\textbf{Lemma}}

\allowdisplaybreaks[4]
\def\BibTeX{{\rm B\kern-.05em{\sc i\kern-.025em b}\kern-.08em
    T\kern-.1667em\lower.7ex\hbox{E}\kern-.125emX}}

\usepackage{marginnote}
\usepackage[normalem]{ulem}
\usepackage[shortlabels]{enumitem}

\begin{document}

\title{Gaming on Coincident Peak Shaving: Equilibrium and Strategic Behavior}

\author{
    Liudong Chen\IEEEauthorrefmark{1}, Jay Sethuraman\IEEEauthorrefmark{2}, Bolun Xu\IEEEauthorrefmark{1} \\
    \IEEEauthorblockA{\IEEEauthorrefmark{1}Earth and Environmental Engineering, Columbia University} \\
    \IEEEauthorblockA{\IEEEauthorrefmark{2}Industrial Engineering \& Operations Research, Columbia University} \\
    \{lc3671, js1353, bx2177\}@columbia.edu  \\
    500 W 120th Street, NYC, NY, 10027, USA 
}


\maketitle

\begin{abstract}

Power system operators and electric utility companies often impose a coincident peak demand charge on customers when the aggregate system demand reaches its maximum. This charge incentivizes customers to strategically shift their peak usage away from the system’s collective peak, which helps reduce stress on electricity infrastructure. In this paper, we develop a game-theoretic model to analyze how such strategic behavior affects overall system efficiency. We show that depending on the extent of customers’ demand-shifting capabilities, the resulting coincident peak shaving game can exhibit concavity, quasi-concavity with discontinuities, or non-concavity with discontinuities. In a two-agent, two-period setting, we derive closed-form Nash equilibrium solutions for each scenario and generalize our findings to multi-agent contexts. We prove the stability of the equilibrium points and propose an algorithm for computing equilibrium outcomes under all game configurations. 
Our results indicate that the peak-shaving outcome at the equilibrium of the game model is comparable to the optimal outcome of the natural centralized model. However, there is a significant loss in efficiency.
Under quasi-concave and non-concave conditions, this inefficiency grows with increased customer flexibility and larger disparities in marginal shifting costs; we also examine how the number of agents influences system performance. Finally, numerical simulations with real-world applications validate our theoretical insights.

\end{abstract}


%
\IEEEpeerreviewmaketitle

\section{Introduction}
Managing peak electricity demand is crucial in power system transmission and distribution due to the need to maintain an instantaneous balance between supply and demand. At the transmission level, total demand must not exceed the maximum generation capacity. At the distribution level, where a tree-branched network of transformer substations steps down voltages, the load at each tier must remain below the capacity of the respective transformers. When demand exceeds the system’s generation or substation capacity, operators may be forced to disconnect portions of the grid to cut off the connecting demand to ensure the secure operation of the remaining system, potentially leading to significant financial losses and adverse impacts on life quality~\citep{power_system}.




With the advent of smart electric meters capable of recording time-dependent electricity usage and the widespread deployment of distributed energy resources (DER), power system operators and electric utility companies are implementing new pricing schemes to promote demand flexibility and alleviate peak demand stress. One such approach is time-of-use (ToU) tariffs~\citep{tou}, in which utilities set hourly electricity prices to encourage customer to shift consumption: increasing usage during low-priced periods (typically at night) and reducing it during high-priced periods (typically during the day). However, ToU tariffs mainly focus on recovering the cost of energy supply and shifting daily consumption, and they are less effective at curbing peak demand over longer periods, such as on a monthly basis, because the prices are not set high enough to further incentivize significant peak demand reduction~\citep{rate_design1,rate_design2}.

Utility companies may also impose direct peak demand charges, billing customers based on their highest power usage with a prefixed peak price over an extended period—typically a month. These charges not only motivate customers to reduce their peak consumption but also help recover distribution grid investment and maintenance costs, as the delivered grid capacity is determined by a customer's peak demand~\citep{peak_charge, peak_demand2}. For instance, Consolidated Edison, serving New York City, may charge more than \$40 per kW of peak demand, with the peak being determined by the highest two consecutive 15-minute intervals in any monthly billing cycle. This effectively translates to a cost of approximately \$80 per kWh during peak periods, which is significantly higher than the standard energy supply rate of around 20 to 30 cents per kWh~\citep{ConEdElectricRates}.


Yet, allocating peak demand charges becomes non-trivial when multiple customers share a capacity bottleneck, such as a constrained substation, because the system's peak does not directly correspond to each individual customer’s peak~\citep{demand_allocation}. Charging a customer for their peak demand may be inequitable if the overall system is not under peak stress, and such charges might not accurately represent the true system peak demand. This issue is particularly critical for utility companies, which are often publicly owned or regulated private monopolies with profit margins capped by their investment costs. Therefore, the revenue from peak demand charges should accurately reflect grid investment expenses based on the actual delivery capacity.

Coincident peak (CP) demand charges bill customers based on their usage during the system's peak period, offering an alternative that more accurately reflects the system condition and the corresponding investment costs. These charges are added to future electricity bills; for example, Texas’s 4CP mechanism applies charges based on a customer’s average demand during the Electric Reliability Council of Texas (ERCOT) system’s peak hour in each of the four summer months (from June to September), with these charges affecting the following year’s bills. Notably, 4CP charges can represent up to 30\% of an organization’s monthly electricity costs~\citep{4CP_ercot}. In 2023, the 4CP charge was approximately \$67/kWh~\citep{4CP_charge_price}, with the total coincident peak demand across the ERCOT region reaching 83.6 GW~\citep{4CP_record}, translating into an estimated \$5.6 billion payments. According to ERCOT’s annual demand response report, the 4CP mechanism facilitated a 3,500 MW demand reduction during peak periods, resulting in savings of approximately \$234.5 million, or 4.2\% of the total 4CP-related costs~\citep{Potomac2023}.

A significant challenge, however, is that the peak period is identified only in hindsight. While numerous academic and industrial efforts have focused on accurately predicting the system’s peak period~\citep{predict_CP_time, baosen_nn, 4CP}, these studies often overlook the interactions among customers, whose collective behavior determines the system peak. This gap naturally motivates a game-based framework and raises a critical question: \emph{Is a game-based framework workable for the CP shaving problem, and how does strategic customer behavior affect system efficiency?}
Addressing this question is the main focus of this paper.

\subsection{Summary of main contributions and implications}
In this work, we formulate a novel game framework for the CP shaving problem. As a first step toward applying a game-theoretic model to address CP shaving challenges, we consider a simplified two-agent, two-period setting, with extensions to more than two agents discussed in Section \ref{multi_section}. The two agents represent clusters of customers or dominant oligopoly customers. The rationale for the two-period setting stems from its ability to represent a peak period and an off-peak period~\citep{jerry}, as identified from historical market data, including demand shifting patterns and price signals.  For example, power system peak demand typically occurs in the early evening, with off-peak periods covering the rest of the day. In the 4CP program in Texas, historical records in 2024 suggest that certain periods, such as around 17:00, are more likely to experience CP events~\citep{4CP_record}. Our model focuses on minimizing system peak demand while deemphasizing off-peak performance, aligning with the primary objective of operators and utilities. This simplified framework offers both insight and tractability for analyzing the fundamental interactions between agents and their impact on system performance. Furthermore, as agents can access the system demand profile from publicly available market information~\citep{load_profile}, they can infer the aggregated response of other agents. We assume that each agent observes the other’s aggregated demand distribution, allowing us to isolate strategic interactions from the complexities of predicting individual valuations. By balancing simplicity with realism, this approach lays a solid foundation for understanding CP shaving strategies and agent behaviors.

With this setting in place, we formulate the CP shaving game model in which agents have a baseline demand at each time period. There is a CP charge price that is announced ahead of time. Each agent has the flexibility to shift some of their demand to earlier or later time periods, and there is a penalty term to account for the cost of this demand-shifting. Each agent's demand-shifting strategy, influenced by baseline demand, penalty costs, and CP charges, determines their consumption in each period. The collective strategies of all agents shape the total consumption in each time period and, hence, the CP period.
Agents with larger baseline demand differences between two periods and lower shifting costs exhibit greater flexibility. Depending on these factors, the model may exhibit concave, quasi-concave with discontinuities, or non-concave with discontinuities characteristics. By contrast, the conventional centralized peak shaving model is formulated as a straightforward concave optimization problem.

We then analytically derive the pure strategy Nash equilibrium (NE) for all CP shaving game structures under the two-agent, two-period setting. Our analysis reveals that the NE is determined by the system's unbalanced demand and the agents' shifting capabilities. By treating the system as a switched dynamic system, we prove that the equilibrium of the game model is globally uniformly asymptotically stable, regardless of switching, provided all agents' demand shifting is within the baseline limit. Furthermore, we demonstrate that a gradient-based algorithm with an update rule serving as a finite difference approximation of the asymptotically stable process can compute the equilibrium point.

We analyze the impact of gaming agents' strategic behaviors by examining the peak shaving effectiveness and the efficiency loss. 
Our findings analytically show that the peak-shaving outcome at the equilibrium of the CP shaving game aligns with the optimal outcome of the natural centralized model but incurs significant efficiency losses—except in the concave game, where the outcomes are fully equivalent. For quasiconcave and non-concave games, we prove that the efficiency loss increases with the disparities among agents, as measured by their marginal shifting costs. Additionally, under identical system conditions, we show that greater agent flexibility exacerbates efficiency loss.

We also generalize our findings to a model with more than two agents, and prove the uniqueness of NE under concave and quasiconcave conditions. For non-concave CP shaving games, we show that system demand can still be balanced across two periods and derive the NE at the group level of agents. We also establish the stability of the equilibrium point and validate the effectiveness of the gradient-based algorithm. The NE reveals equivalent peak shaving effectiveness between the game and the centralized model while showing that changes in the number of agents influence the game type and, consequently, the efficiency loss.

For industry practitioners and policymakers, our paper offers the following takeaways:
\begin{itemize}
    \item We show that the game model has a stable NE that can be computed efficiently, and applying this model to the CP shaving problem achieves peak shaving effectiveness comparable to that of an efficient centralized model.
    Thus, operators can adopt the game model without concerns about its impact on CP shaving effectiveness.
    \item We identify that system's inefficiencies in game environments stem from agents' greater flexibility and disparities in their marginal shifting costs. This insight guides operators and policymakers in designing targeted incentive mechanisms tailored to agents, enhancing both system efficiency and equity.
    \item We show that the economic performance of small systems with fewer agents gaming on CP shaving is sensitive to agents' flexibility. Therefore, one should select flexible agents to form larger systems while leaving inflexible agents to form small gaming systems.
\end{itemize}

\subsection{Related work}
To precisely position our study within the existing literature, we elaborate on our work from three perspectives closely related to our CP shaving problem: (i) Load management, (ii) CP demand charge, and (iii) Related game formulations.

\textbf{Load management.}
Load management is a border concept covering peak demand reduction, load shifting, and demand response. Many studies, from the utility's perspective, formulate optimization problems to design pricing mechanisms - such as tariffs or peak demand charges - aimed at minimizing convex costs or maximizing concave profits~\citep{DR}. These pricing design approaches often rely on customers' price response behavior, typically leading to a two-layer framework (or Stackelberg game model): the upper layer determines the pricing strategy, while the lower layer models demand based on customer response behavior~\citep{two_layer}. The lower layer may involve constructing utility functions to represent customer preferences~\citep{utility, behavior} or applying data-driven methods to learn behavior from historical price and consumption data~\citep{data_hard}. From a demand-side perspective, the key challenge lies in distinguishing between baseline demand and observed demand to better understand customers' demand reduction decision-making behavior~\citep{demand_change}. Beyond individual customers, the interactive influence of demand strategies has been widely studied through non-cooperative game formulations and generally achieves peak shaving~\citep{noncooperative_game}. However, while these load management approaches contribute to peak demand reduction, they primarily focus on maximizing the profit of different market participants rather than directly addressing CP demand reduction in accordance with CP charging mechanisms.


\textbf{CP demand charge.}
Academic studies and industrial solutions for CP shaving generally focus on predicting the timing of the system peak and making decisions from the customers' perspective. Prediction involves estimating the CP period as a probabilistic distribution, often using machine learning tools that leverage input features such as historical demand and weather conditions~\citep{prediction1, baosen_nn}. Following this, decisions are made based on the predictions.
Industrial solutions primarily emphasize short-term prediction, employing auto-regressive methods to iteratively update models and issue warning signals to customers~\citep{4CP}. Academic solutions typically frame the problem as a scheduling task, making deliberate decisions based on CP period distributions. Examples include stochastic sequential optimization~\citep{adam_datacenter} and optimization with neural networks trained as decision policies~\citep{baosen_cdc}.
Recently, an approach combining prediction and decision-making, termed 'decision-focused learning' or contextual optimization~\citep{contextual}, has gained traction. This method trains weighting parameters using decision losses from downstream optimization rather than prediction losses, aligning with the optimization tasks to produce effective decisions~\citep{spo}. 
However, it has yet to be applied to the CP shaving problem. Moreover, existing studies often assume a single agent or 'price-taker' condition, neglecting the impact of various agents' decisions on CP period predictions. This motivates us to incorporate customer interactions as a critical aspect of CP shaving design.

\textbf{Related game formulations.}
Without specific constraints, the CP shaving problem can be abstracted as agents betting on discrete events, with winners sharing benefits based on their bids. This framework mirrors the classic concept of pari-mutuel betting, in which the odds (or CP periods) fluctuate in response to all participants' betting strategies~\citep{betting1, pari_mutuel}. Pari-mutuel betting is commonly applied in sports, where optimal betting strategies—both in size and target—are determined by solving discrete decision problems aimed at maximizing expected returns through the best combination of picks~\citep{betting_pool1, pari_mutuel1}.

The CP shaving problem also shares similarities with congestion games by treating peak demand consumption as a form of resource usage, where increased utilization by multiple agents results in higher costs~\citep{congestion_game}. Congestion games can often be formulated as potential games, guaranteeing the existence of pure-strategy Nash equilibria (NE) even in weighted settings, provided certain conditions on the utility functions are met~\citep{congestion_game_weight}. The primary difference between congestion games and CP shaving is in their cost structures: in congestion games, costs are determined by the aggregate resource usage, whereas in CP shaving, costs are driven by the peak-period demand—a single, critical resource. This makes the problem more analogous to a bottleneck congestion game, where the objective is to optimize the performance of the worst-performing element~\citep{Bottleneck_game}. In such settings, pure-strategy NE has also been shown to exist,  though they often result in an inefficient price of anarchy - exceeding one~\citep{bottleneck_PoA}.

Despite the valuable insights provided by these game-theoretic frameworks, the CP shaving problem exhibits key differences. Specifically, demand shifting can negatively impact customer comfort and lead to profit losses, the cost structure is dictated by peak-period demand rather than the maximum cost across periods, and the overall system demand must remain unchanged. 
We integrate these considerations into the CP shaving game framework developed and analyzed in this paper.



\section{Model and Preliminaries}
In this section, we formulate the CP shaving game model and introduce some definitions.
We begin with a two-agent (x and y), two-period (1 and 2) model, where the two agents can represent clusters of customers with different baseline demand conditions, and the two periods correspond to a CP period and an off-CP period, respectively. 
We let agent x's baseline demand in periods 1 and 2 be $X_1$ and $X_2$, respectively. Similarly, agent y's baseline demand in periods 1 and 2 are, respectively, $Y_1$ and $Y_2$. The system's baseline demand in periods 1 and 2, obtained by summing both agents’ baseline demands in that period, are denoted $S_{\mathrm{b},1}$ and $S_{\mathrm{b},2}$ respectively. Without loss of generality, we assume the baseline demand in period 2 is greater than in period 1, i.e., $S_{\mathrm{b},2} > S_{\mathrm{b},1}$. In the event that is not true for the given instance, we can interchange the roles of the two periods, and all the analysis will follow. Agents have the flexibility to shift some of their consumption from one period to the other. We let $x$ and $y$ denote, respectively, the amount of demand shifted by agents x and y, respectively, {\em from} period 2 {\em to} period 1. In other words, for any given $x$, the demand of agent x in the two periods are, respectively, $X_1 + x$ and $X_2 - x$; similarly for $y$ and agent y. Note that, without loss of generality, we further assume that agent x's baseline demand in period 2 is greater than in period 1. i.e., $X_2 > X_1$. This ensures that agent x's demand shifting $x$ is non-negative at the equilibrium, while agent y's shifting $y$ remains unrestricted. Still, the roles of agents x and y are interchangeable, so all results remain valid. The system’s total demand, which includes both agents' demand shifting, is denoted as $S_1$ and $S_2$ in periods 1 and 2, respectively.
Depending on the agents' demand-shifting behaviors, the roles of the CP and off-CP periods can change, i.e., either period 1 or period 2 can become the CP or off-CP period. We use $\alpha_{\mathrm{x}},\alpha_{\mathrm{y}}\in\mathbb{R}^{+}$ to denote the penalty parameter for agent x and y, representing the comfort loss or perceived cost of shifting demand. These parameters can be identified from shifting behavior suggested by historical data.
The fixed CP price is denoted by $\pi\in \mathbb{R}^{+}$ and $I(x,y)$ is the step function (indicator function). The agents choose $x$ and $y$ to maximize their payoffs. Formally, the game model is defined as $G=(N,\mathcal{X},U)$, where
\begin{itemize}
    \item $N=\{\mathrm{x},\mathrm{y}\}$ is the two-agent agent set;
    \item $\mathcal{X}=\mathcal{X}_{\mathrm{x}} \times \mathcal{X}_{\mathrm{y}}$ is the strategy set formed by the product topology of each agent's strategy set;  
    \item  $U =\{f_{\mathrm{x}},f_{\mathrm{y}}\}$ is the payoff function set, where $f_{\mathrm{x}},f_{\mathrm{y}}: \mathcal{X} \rightarrow \mathbb{R}$.
\end{itemize}
During the game, agent x chooses strategy $x \in \mathcal{X}_{\mathrm{x}}$ to maximize its payoff $f_{\mathrm{x}}(x,y)$ with $y \in \mathcal{X}_{\mathrm{y}}$ chosen by agent y. The payoff to agent x is simply the negative of its costs over the two time periods, and is given by:
\begin{subequations} \label{agenta}
    \begin{align}
        \max_{x} &f_{\mathrm{x}} (x,y)= -\pi(X_{1}+x) I(x,y)
        -\pi(X_{2}-x)(1- I(x,y))
        -\alpha_{\mathrm{x}} x^2, \label{cost_game}\\
        &I(x,y) = \begin{cases} 
      1 & S_1(x,y) - S_2(x,y)\geq 0 \\
      0 & S_1(x,y) - S_2(x,y) < 0 
        \end{cases},\label{step_function} \\
        &S_1(x,y) = X_{1}+Y_{1}+x+ y = S_{\mathrm{b},1}+x+ y, \\
        &S_2(x,y) = X_{2}+Y_{2}-x-y = S_{\mathrm{b},2}-x -y.
    \end{align}
\end{subequations}
Agent y's payoff function $f_{\mathrm{y}}(x,y)$ follows the same structure by changing baseline demand to $Y_1,Y_2$, demand shifting to $y$, and penalty parameter to $\alpha_\mathrm{y}$. Note that there is always a CP period (and a corresponding charge) even if demands in the two periods are equal; in that case, we treat period 1 as the CP period, as reflected in our definition of the step function $I(\cdot)$ (\ref{step_function}). Also, we focus on demand shifting and assume a constant energy rate besides the peak demand charge, so we do not consider the energy cost.


We next introduce the definitions of equilibrium and continuity.  


\begin{definition}\label{Nash_Equ_defi}
We define the following:
(1) \emph{Pure-strategy Nash equilibrium (\citep{nash1950}).} 
    $(x^*,y^*)\in \mathcal{X}$ is a NE in pure strategies of the game $G$ if and only if (iff) $f_{\mathrm{x}}(x^*,y^*) \geq f_\mathrm{x}(x,y^*)$ and $f_{\mathrm{y}}(x^*,y^*) \geq f_\mathrm{y}(x^*,y)$ for every $x\in \mathcal{X}_{\mathrm{x}}, y \in \mathcal{X}_{\mathrm{y}}$. 
    
(2) \emph{Upper semi-continuity (u.s.c.) and lower semi-continuity (l.s.c.).} 
    The function $f: \mathcal{X} \rightarrow \mathbb{R}$ is called u.s.c. or l.s.c. if for every $x_0$ such that 
    \begin{align}
        &\limsup\limits_{x\rightarrow x_0} f(x)\leq f(x_0), 
        \text{ or } \liminf\limits_{x\rightarrow x_0} f(x)\geq f(x_0),
    \end{align} for all $x$ in some neighborhood of $x_0$, respectively.
\end{definition}


The step function makes the game analysis non-trivial. We separate the overall system into two periods corresponding to CP period in 1 and 2, i.e., $\mathcal{X}_{\mathrm{cp}1}=\{x,y|S_1(x,y)\geq S_2(x,y)\}$ and $\mathcal{X}_{\mathrm{cp}2}=\{x,y|S_1(x,y) < S_2(x,y)\}$, with the payoff functions, take agent x for example,
\begin{subequations}\label{1a}
    \begin{align}
        f_{\mathrm{x},1}(x) = -\pi(X_{1}+x)-\alpha_\mathrm{x} x^2, x\in \mathcal{X}_{\mathrm{cp}1}, \nonumber\\
        f_{\mathrm{x},2}(x) = -\pi(X_{2}-x)-\alpha_\mathrm{x} x^2,x\in \mathcal{X}_{\mathrm{cp}2}, \label{objective}
    \end{align}
Note that each payoff function is concave, and by applying the first-order optimality condition~\citep{convex_optimization}, we have 
\begin{align}
     x'= \arg\max_{x} f_{\mathrm{x},1}(x) = -\frac{\pi}{2\alpha_\mathrm{x}},  x\in \mathcal{X}_{\mathrm{cp}1}, \nonumber \\
     x'= \arg\max_{x} f_{\mathrm{x},2}(x) = \frac{\pi}{2\alpha_\mathrm{x}}, x\in \mathcal{X}_{\mathrm{cp}2}, 
\end{align}
\end{subequations}

With these results, we define critical points, balance points, and system average demand in the game model.
\begin{definition}\emph{Critical Points, balance points, and system average demand.}\label{Special_points}
    \begin{subequations}
   We define the following:
    \begin{enumerate}
        \item \emph{Critical points} for agents x and y:
        \begin{align}
            r_\mathrm{x} = \frac{\pi}{2\alpha_\mathrm{x}}, \quad r_\mathrm{y} = \frac{\pi}{2\alpha_\mathrm{y}}.
        \end{align} 
        \item \emph{Agent balance points} for agents x and y:
        \begin{align}
            b_\mathrm{x} =\frac{X_{2}-X_{1}}{2}, \quad b_\mathrm{y} =\frac{Y_{2}-Y_{1}}{2}.
        \end{align}
        \item \emph{System balance point}:
        \begin{align}
            b =\frac{S_{\mathrm{b},2}-S_{\mathrm{b},1}}{2}.
        \end{align}
        \item \emph{System average demand}:
        \begin{align}
            S = \frac{S_{\mathrm{b},2}+S_{\mathrm{b},1}}{2}.
        \end{align}
    \end{enumerate}
    \end{subequations}
\end{definition}

The critical points $r_{\mathrm{x}}, r_{\mathrm{y}}$ are determined by the payoff function in period 2, as exemplified in $f_{\mathrm{x},2}$ (\ref{1a}). These points represent the optimal demand shifting for each agent within period 2 to maximize its payoff. Naturally, the corresponding optimal shifting amounts for period 1 are $-r_{\mathrm{x}}, -r_{\mathrm{y}}$. The agent balance points $b_\mathrm{x},b_\mathrm{y}$ represent inter-period demand shifting. For instance, if agent x chooses to shift $b_x$ from period 2 to period 1, its demand in both periods will be identical, and the same applies to agent y. Furthermore, if the total shifted demand across both agents equals $b$, then the system demand in periods 1 and 2 will be balanced and equal to $S$, which is the system average demand. Given the assumption that $S_{\mathrm{b},1} < S_{\mathrm{b},2}$ and $X_1<X_2$, it follows that $b_\mathrm{x}>0$ and $b>0$.


We then define capable and non-capable agents using agent x as an example, with agent y following the same definition.
\begin{definition}\emph{Capable and non-capable agents.}\label{capably}
    According to Definition \ref{Special_points}, given agent x with baseline demand $X_{1},X_{2}$, CP price $\pi$, and penalty parameter $\alpha_\mathrm{x}$, we define the agent x is \emph{capable} if it satisfies
    \begin{align}
        -r_\mathrm{x} \leq b_\mathrm{x} \leq r_\mathrm{x}. \label{concave_condition}
    \end{align}
    Define an agent who does not satisfy (\ref{concave_condition}) as a \emph{non-capable} agent, including \emph{upper non-capable} agent with $b_\mathrm{x} > r_\mathrm{x}$ and \emph{lower non-capable} agent with $b_\mathrm{x} < -r_\mathrm{x}$.
\end{definition}

According to Definition \ref{Special_points}, (\ref{concave_condition}) means agent x is economically capable of balancing its demand in the two periods to reduce the CP charge. This is because demand shifting will first reach the balance point before reaching the critical point. If agent x cannot balance its demand across two periods, its maximum shifting capability is constrained by the optimal shifting amount within a single period. Thus, the maximum shifting capability of agent x is given by $\min\{r_\mathrm{x}, b_\mathrm{x}\}$. Similarly, for agent y, the maximum shifting capability is $\min\{r_\mathrm{y}, b_\mathrm{y}\}$ if $b_\mathrm{y}>0$, or $\max\{-r_\mathrm{y},b_\mathrm{y}\}$ if $b_\mathrm{y}<0$.

\section{Equilibrium analysis of the CP shaving game}
In this section, we first show the properties of the two-agent two-period CP shaving game $G$ with the specific agent type and prove the pure-strategy NE of the game. 


\subsection{Game properties}
We first show the property of CP shaving game $G$. Determined by the parameters of all agents, the game performs differently. 
\begin{proposition}\emph{Concave, quasiconcave/discontinuous, and non-concave/discontinuous CP shaving game.} \label{game_proper}
    Given critical point $r_\mathrm{x},r_{\mathrm{y}}$ and balance point $b_\mathrm{x},b_{\mathrm{y}},b$ as define in Definition \ref{Special_points}, the two-agent two-period CP shaving game $G$ satisfies one of the following:
\begin{subequations}
    \begin{enumerate}
        \item \emph{Concave CP shaving game.} The game $G$ is concave if all agents' payoff functions $f_\mathrm{x}(x,y)$ ($f_\mathrm{y}(x,y)$) are concave in $x\in \mathcal{X}_\mathrm{x}$ ($y\in \mathcal{X}_\mathrm{y}$) for each $y\in \mathcal{X}_{\mathrm{y}}$ ($x\in \mathcal{X}_{\mathrm{x}}$). This concavity condition is satisfied if the following holds
        \begin{align}
            b &> r_\mathrm{x} + r_{\mathrm{y}}.\label{concave}
        \end{align}
        \item \emph{Quasiconcave CP shaving game.} The game $G$ is quasiconcave/discontinuous if all agents' payoff functions $f_\mathrm{x}(x,y)$ ($f_\mathrm{y}(x,y)$) are quasiconcave in $x\in \mathcal{X}_\mathrm{x}$ ($y\in \mathcal{X}_\mathrm{y}$) for each $y\in \mathcal{X}_{\mathrm{y}}$ ($x\in \mathcal{X}_{\mathrm{x}}$)~\citep{reny1999}. This concavity condition is satisfied if the following holds
    \begin{align}
        -r_\mathrm{x} \leq b_\mathrm{x} \leq r_\mathrm{x},\ -r_\mathrm{y} \leq b_\mathrm{y} \leq r_\mathrm{y};\label{quasi_discounti}
    \end{align}
    \item \emph{Non-concave CP shaving game.} The game $G$ is non-concave/discontinuous if it is not concave or quasiconcave. The condition is
    \begin{align}
       \{0 < b \leq r_\mathrm{x}+r_{\mathrm{y}} \}
       \cap &[\{b_{\mathrm{x}} > r_{\mathrm{x}} \}\cup \{b_{\mathrm{y}} < -r_{\mathrm{y}}\} \cup \{b_{\mathrm{y}} > r_{\mathrm{y}}\}]. \label{non-concave}
    \end{align}
    \end{enumerate}
\end{subequations}
        
       
     
    
\end{proposition}


\begin{proof}[Sketch of the proof]
The payoff function is quadratic when there is only one possible CP period in the gaming process, resulting in a concave and continuous game. The indicator function discontinues the game when the CP period changes. According to the quasiconcavity definition, we separate three cases: 1) Agent x can switch the CP period regardless of whether it occurs in period 1 or 2; 2) Agent x can switch the CP period only if it occurs in period 1, and its payoff function is l.s.c. as defined in Definition \ref{Nash_Equ_defi}; 3) Agent x can switch the CP period only if it occurs in period 2, and the function is u.s.c. as defined in Definition \ref{Nash_Equ_defi}. We show only the first case is able for both agents' payoff functions to be quasiconcave. Then the non-concave condition is derived from the complementary set of concave and quasiconcave conditions. 
The detailed proof is provided in the appendix.
\end{proof}

This Proposition shows the game's properties depend on the relationship between the agent's critical point and balance point, affected by baseline demand $X_{1},X_{2},Y_{1},Y_{2}$, CP price $\pi$, and shifting penalty parameters $\alpha_\mathrm{x},\alpha_{\mathrm{y}}$. Under concave game conditions, the CP period always aligns with the baseline since agents' demand shifting is insufficient to alter it. If all agents are capable, as defined in Definition \ref{capably}, the CP shaving game is quasiconcave. Otherwise, if both agents' demand shifting is sufficient to interchange the CP period but at least one agent is non-capable, the game becomes non-concave. In practical systems, due to the variant agents' shifting capability, all three types of games are likely to appear. Empirical studies on demand response indicate that some industrial customers with automation can shift demand efficiently, while some are inefficient in shifting~\citep{DR}. 
With these properties of CP shaving game $G$, we then analyze its NE.



\subsection{Nash Equilibrium}
In this section, we study the NEs of the two-agent two-period CP shaving games in all conditions as described in Proposition \ref{game_proper}. The main theorem is as follows.
\begin{theorem}\emph{NEs in two-agent two-period CP shaving game.} \label{n_agent_game}
Given critical points $r_{\mathrm{x}},r_{\mathrm{y}}$ and balance points $b_{\mathrm{x}},b_{\mathrm{y}}, b$ as defined in Definition \ref{Special_points}, the unique pure-strategy NE $(x^*,y^*)$ as defined in Definition \ref{Nash_Equ_defi} of the two-agent two-period CP shaving game $G$ satisfy one of the following: 
\begin{subequations} \label{general_game_results}
\begin{enumerate}
    \item \emph{Concave CP shaving game}. 
    \begin{align}
    x^* &= r_{\mathrm{x}},\ y^* = r_{\mathrm{y}}, 
      \label{fully_coopera2} 
    \end{align}
    \item \emph{Quasiconcave CP shaving game}.
    \begin{align}
        x^* &= b_{\mathrm{x}},\ y^* = b_{\mathrm{y}}, 
      \label{fully_compet1}
    \end{align}
    \item \emph{Non-concave CP shaving game}.
    \begin{align}
    x^* &= r_{\mathrm{x}},
    y^* = b - r_{\mathrm{x}}, &\text{if }
         \{0<b \leq r_{\mathrm{x}} + r_{\mathrm{y}}\} \cap \{b_{\mathrm{x}} > r_{\mathrm{x}}\},\nonumber\\
         x^* &= b+r_{\mathrm{y}},
        y^* = -r_{\mathrm{y}}, &\text{if }
         \{0<b \leq r_{\mathrm{x}} + r_{\mathrm{y}}\} \cap \{b_{\mathrm{y}} < -r_{\mathrm{y}}\},\nonumber\\
         x^* &= b - r_{\mathrm{y}},
        y^* = r_{\mathrm{y}}, &\text{if }
         \{0<b \leq r_{\mathrm{x}} + r_{\mathrm{y}}\} \cap \{b_{\mathrm{y}} > r_{\mathrm{y}}\}.
     \label{mix2}   
    \end{align}
\end{enumerate}
\end{subequations}
\end{theorem}

\begin{proof}[Sketch of the proof]
We first provide a Lemma (Lemma \ref{lemma_NE2agent}) to establish the NE for the two-agent, two-period quasiconcave CP shaving game $G$, providing insight into the best response rationale. Then we separate two possible cases depends on agents' baseline demand conditions, then study whether both agents are capable, one of them is capable, and none of them is capable. The reason is that agents' best response is to shift demand away from the CP period when their individual demand is higher in the CP period, to avoid CP charges, or shift demand toward the CP period when their individual demand is lower in the CP period, to maintain the CP period. However, the shifting capability is determined by the relationship between the critical point $r_{\mathrm{x}}$ and the agent balance point $b_{\mathrm{x}}$.
We analyze each condition and show corresponding NE solutions. The details are provided in the appendix.
\end{proof}



This theorem shows the NEs and correspind conditions under each concave, quasiconcave, and non-concave CP shaving game type. 
The NE describes the relationship of gaming agents as fully cooperative (\ref{fully_coopera2}), fully competitive (\ref{fully_compet1}), and mixed competitive and cooperative (\ref{mix2}). Specifically, When they are in a fully cooperative relationship, they cooperate to shift the demand away from the CP period. Once they can change the CP period together, they enter a competitive relationship where they compete with each other to change the CP period, but one agent's decision is limited by their shifting capability. When all agents are capable, they are in a fully competitive relationship that can balance their demand regardless of opponents' strategy.
Indeed, if both agents want to be in a mixed relationship, they must be 'asymmetrical,' where their shifting penalty parameters should be different so that they have different shifting limitations. Otherwise, they are either fully cooperative or fully competitive. Note that the condition described in this theorem aligns with the conditions in Proposition \ref{game_proper}, specifically for non-concave game, we write out each condition conditions along with the NE solutions.

We also demonstrate that, in practice, under quasiconcave and concave game conditions, agents share the market equally, meaning their strategies depend solely on their own parameters, independent of their opponent's strategy. In contrast, in the non-concave game, an agent's strategy is influenced by the opponent's actions. For example, when one agent's strategy follows $y=b \mp r_{\mathrm{x}}$ it tends to be more flexible and dominates the other, less flexible agent with $x = \pm r_{\mathrm{x}}$, gaining a competitive advantage and benefiting more. Notably, agents' demand-shifting strategies $x^*,y^*$ reflect their flexibility and are determined by the game type as well as the shifting penalty parameter $\alpha_{\mathrm{x}},\alpha_{\mathrm{y}}$ and baseline demand conditions $b_{\mathrm{x}},b_{\mathrm{y}}$ (depends on $X_{1},X_{2},Y_{1},Y_{2}$).


\section{Equilibrium stability and algorithm convergence}
The CP shaving game system is a switched dynamics system due to the indicator function in the model (\ref{agenta}). In this section, we first analyze the system stability and then develop a solution algorithm to reach the stable (equilibrium) point.
\subsection{Equilibrium stability}
There are two system dynamics corresponding to periods 1 and 2, separated by the indicator function and a switching logic between these two periods. According to~\citep{rosen1965}, we consider a reasonable dynamic model in each period in which each agent changes his strategy following the gradient direction with respect to his strategy of his payoff function, then each agent's payoff will increase given all other agents' strategies. Denote the dynamic time index with $k$, and the gradient as $F_{1}(x,y) = [\nabla_\mathrm{x} f_{\mathrm{x}}(x,y),\nabla_{\mathrm{y}}f_{\mathrm{y}}(x,y)]^T$ with total differential operator $\nabla$, the dynamics of period 1 ($x,y\in \mathcal{X}_{\mathrm{cp}1}$) is 
\begin{subequations}
\begin{align}
    \frac{d[x,y]^T}{dk} &= [\dot{x},\dot{y}]^T = F_1(x,y)= [-(\pi+ 2\alpha_{\mathrm{x}}x),
        -(\pi+2\alpha_{\mathrm{y}}y)]^T.\label{dynamics_system1}
\end{align} 
The period 2, i.e., $x,y\in \mathcal{X}_{\mathrm{cp}2}$, follow the same structure with different gradient $F_{2}(x,y)$:
    \begin{align}
    \frac{d[x,y]^T}{dk} &= [\dot{x},\dot{y}]^T= F_2(x,y) = [\pi- 2\alpha_{\mathrm{x}}x,
        \pi-2\alpha_{\mathrm{y}}y]^T.\label{dynamics_system2}
    \end{align}
\end{subequations}

We then denote the CP shaving game system with the switching logic as 
\begin{align}
    [\dot{x},\dot{y}]^T = \begin{cases} 
      F_1(x,y) & x,y\in \mathcal{X}_{\mathrm{cp}1} \\
      F_2(x,y) & x,y\in \mathcal{X}_{\mathrm{cp}2}
        \end{cases}. \label{switched_system}
\end{align}

Our goal is to prove the overall CP shaving game system with the switching logic is global uniform asymptotically stable~\citep{switching_dynamics}. According to the NE described in Theorem \ref{n_agent_game}, we call $(x^*,y^*)$ as the equilibrium point of the CP shaving game system (\ref{switched_system}) if one of the following is satisfied,
\begin{subequations}\label{equilbrium_point}
\begin{align}
    f_{\mathrm{x},1}(x^*) + f_{\mathrm{y},1}(y^*) &= f_{\mathrm{x},2}(x^*) + f_{\mathrm{y},2}(y^*), \label{1NE}\\
    F_2(x^*,y^*) &= 0.  \label{3NE}
\end{align}
\end{subequations}
Among them, the first condition (\ref{1NE}) corresponds to non-concave CP shaving game (\ref{mix2}) and quasiconcave CP shaving game (\ref{fully_compet1}), where both periods have the same CP charges and shifting penalty because the system demand is balanced in the two periods at NE and the shifting penalty is symmetric in the two periods. The second (\ref{3NE}) conditions correspond to the concave CP shaving game (\ref{fully_coopera2}), where the NE is obtained when the gradient $F_2$ reaches zero.
Then, we introduce the main stability results.
\begin{theorem}\emph{Global stability of equilibrium point.}\label{dynamics_theorem}
    The CP shaving game system (\ref{switched_system}) is global uniform asymptotically stable in $\mathcal{X}_{\mathrm{s}}$, where
        \begin{align} \label{local_stable}
          \mathcal{X}_{\mathrm{s}} = \{x,y|  \alpha_{\mathrm{x}} x^2 + \alpha_{\mathrm{y}} y^2 + \pi(S_{\mathrm{b},1} + x+y) > 0 
            \cap \alpha_{\mathrm{x}} x^2 + \alpha_{\mathrm{y}} y^2 + \pi(S_{\mathrm{b},2}-x-y) > 0\},
        \end{align}
    i.e., for every starting point $(x,y) \in \mathcal{X}_{\mathrm{s}}$, the solution $(x(k),y(k))$ to the CP shaving game system (\ref{switched_system}) converges to an equilibrium point $(x^*,y^*)\in \mathcal{X}_{\mathrm{s}}$ as $k \rightarrow \infty$, where $(x^*,y^*)$ is the equilibrium point satisfy (\ref{equilbrium_point}).
\end{theorem}
\begin{proof}[Sketch of the proof]
The switched dynamics system property makes the stability analysis non-trivial. We first show that the dynamical systems within each period described by (\ref{dynamics_system1}) and (\ref{dynamics_system2}) are asymptotically stable in $\mathcal{X}$.
Then, we include the switching logic and derive multiple Lyapunov functions with continuous properties in the switching surface to show the CP shaving game system (\ref{switched_system}) is global uniform asymptotically stable in $\mathcal{X}_{\mathrm{s}}$~\citep{switching_dynamics}.
After that, we show the stable point is obtained with $F_2=0$ under the condition of concave CP shaving game, and with $f_{\mathrm{x},1}+f_{\mathrm{y},1} = f_{\mathrm{x},2}+f_{\mathrm{y},2}$ under the condition of quasiconcave or non-concave CP shaving game. From Theorem \ref{n_agent_game}, we also know the equilibrium is unique, which means the stable point is the equilibrium point. The detailed proof is provided in the appendix. 
\end{proof}

This Theorem shows the global uniform asymptotically stability of our CP shaving game system in the strategy set $\mathcal{X}_{\mathrm{s}}$. Combine with Theorem \ref{n_agent_game}, the CP shaving game $G$ can converge to the unique equilibrium points for the solution trajectory that within $\mathcal{X}_{\mathrm{s}}$. It is important to know that $\mathcal{X}_{\mathrm{s}}$ is reasonable in real operations. Specifically, the key component of $\mathcal{X}_{\mathrm{s}}$ is $\pi(S_{\mathrm{b},2}-x-y)$, which means when agents shift demand, the shifting amount remains within the system’s baseline limitation. This holds naturally and ensures practical system stability, as demand shifting cannot result in negative system demand, meaning agents cannot "borrow" demand to shift beyond their available capacity.

With the stability property, our next step is to develop an algorithm to compute the equilibrium point.
 
\subsection{Algorithm to determine the equilibrium point}
Given the finite difference approximation to the CP shaving game system dynamics (\ref{switched_system}) with learning rate vector $\tau_{\mathrm{x},h},\tau_{\mathrm{y},h}$ for agent x and y, as well as gradient $F_j$, we have
        \begin{align}
            [x_{h+1},y_{h+1}]^T = [x_{h},y_{h}]^T + \text{diag}(\tau_{\mathrm{x},h},\tau_{\mathrm{y},h}) F_{j}(x_{h},y_{h}), j=1,2,  \label{gradient_algorithm}
        \end{align}
where $\text{diag}(\cdot):\mathbb{R}^2 \to \mathbb{R}^{2 \cdot 2}$, $j$ is the switching signal taking the value 1 for $S_1(x,y)\geq S_2(x,y)$ and 2 for $S_1(x,y)<S_2(x,y)$.
This then forms a gradient-based algorithm following the updating rule (\ref{gradient_algorithm}) to gradually reach the equilibrium point. Among them, the gradient is determined by the period on which the current solution lies, and the learning rate for agents x and y $\tau_{\mathrm{x},h},\tau_{\mathrm{y},h}$ are determined by its payoff function and gradient. Choosing a suitable learning rate is the key to showing convergence performance, we thus provide the following Theorem.
\begin{theorem}\emph{Determination of the equilibrium point.}  \label{convergency_theorem} 
Given the finite difference approximation as described in (\ref{gradient_algorithm}), a finite learning rate $\tau_{\mathrm{x},h}$ for agent x can be selected such that when $X_{1}+x\geq X_{2}-x$
\begin{subequations}\label{converge_condition}
    \begin{align}
        -f_{\mathrm{x},1} (x_{\overline{h}}) < -f_{\mathrm{x},1}(x_{\underline{h}}), 
        -f_{\mathrm{x},2} (x_{\overline{h}}) > -f_{\mathrm{x},2}(x_{\underline{h}}),
    \end{align}
    when $X_{1}+x < X_{2}-x$,
    \begin{align}
        -f_{\mathrm{x},1} (x_{\overline{h}}) > -f_{\mathrm{x},1}(x_{\underline{h}}), 
        -f_{\mathrm{x},2} (x_{\overline{h}}) < -f_{\mathrm{x},2}(x_{\underline{h}}),
    \end{align} 
\end{subequations}
where $F_j(x_h,y_h) \neq 0, j=1,2$; $\underline{h}$ and $\overline{h}$ is a switched pair that satisfies $\underline{h}<h<\overline{h}$ and $\underline{h} = \overline{h} = j$, $h \neq j$. The same also holds for the learning rate $\tau_{\mathrm{y},h}$ of agent y.
\end{theorem}

\begin{proof}[Sketch of the proof] 
We use the backtracking line search method~\citep{convex_optimization} to calculate the learning rate, which is affected by the CP periods of the current and future steps. For the concave game, only the baseline CP period (2) could be the CP period during the entire solution process, according to the backtracking line search, the gradient $F_2$ reduce gradually, and combined with Theorem \ref{dynamics_theorem}, the gradient will reduce to zero and reach the equilibrium point satisfy (\ref{3NE}). 

When switching happens, agents' individual peak periods are different, and we analyze the objective function $-f_{\mathrm{x},1},-f_{\mathrm{x},2}$ change following the backtracking line search criteria by judging where the CP periods is in the current and next steps. Given a trajectory starts from period 1 at $h$, switch to period 2 at $h+1$, and back to period 1 at $h+2$, suppose agent x's individual peak period is 1, we then show agent x updates its decision $x$ when switching from period 1 to 2, and agent y update its decision $y$ when switching from period 2 to 1, i.e., the $\tau_{\mathrm{y},h}=0, \tau_{\mathrm{x},h+1}=0$, correspondingly, agent x's objective $-f_{\mathrm{x},1}$ decrease and agent y's objective $-f_{\mathrm{y},2}$ decrease. The reason for the change is that their individual peak period aligns with the system's CP period. Note that this learning rate is determined by their payoff functions and gradients, allowing the gradient for agent x to increase its decision variable while the gradient for agent y keeps the decision variable. By the same analysis, we know agent x's objective $-f_{\mathrm{x},2}$ increase and agent y's objective $-f_{\mathrm{y},1}$ increase in the trajectory starting from period 2, switched to period 1, and back to period 2. We then show these results still hold when the trajectory stays more steps in the period that between two transitions. This proves the reduction of the distance between the objective functions in the two periods. Combined with Theorem \ref{dynamics_theorem}, the entire system finally will satisfy (\ref{1NE}) and reach the equilibrium point.
The details are provided in the appendix.
\end{proof}

This Theorem shows that the finite learning rate can be chosen in each step using the backtracking line search. Following this updating rule (\ref{gradient_algorithm}), combined with Theorem \ref{dynamics_theorem}, a gradient-based algorithm can reach the equilibrium point. Specifically, the difference in each agent's payoff function in the two periods reduces gradually, and depending on whether the game is concave or not, the difference can be reduced to zero or until baseline CP periods' gradients reaches zero. 

As Theorem \ref{dynamics_theorem} shows, the system is asymptotically stable in $\mathcal{X}_{\mathrm{s}}$. Combined with Theorem \ref{convergency_theorem}, this confirms that the algorithm computes the NE as described in Theorem \ref{n_agent_game}. When applying this game framework to a practical CP shaving system, the convergence process naturally reflects real agent interactions in the gaming process. Specifically, the system dynamics alternate between periods 1 and 2, while agents iteratively update their decisions based on the best response strategy. That is, given the opponents' (aggregated) strategy, each agent determines an optimal response based on its own payoff. This iterative process follows the system dynamics and ultimately converges to the NE described in Theorem \ref{n_agent_game}.
Note that we don't use a higher-order gradient descent method, such as Newton's method, because the fast updating with higher-order gradient information lets the solution in each period converge to each period's stable point too fast to realize the converge on the overall switched system.

\section{Impact of Customers Strategic Behavior}
In this section, we analyze gaming agents' strategic behavior in the two-agent two-period setting. First, we state a benchmark centralized CP shaving model for comparison, then analyze agents' strategic behavior from both an economic perspective with the efficiency loss, and a technical perspective with the peak shaving effectiveness. 

\subsection{Centralized CP shaving}
We first introduce the centralized CP shaving model, which assumes a central operator has direct control over both agents in the formulated two-period setting to minimize their total cost. This model represents a centralized realization of the CP shaving game. Consequently, the centralized model maximizes the total objective function of both agents and is formulated as follows
\begin{align}
    (x^*,y^*) \in\arg \max_{x, y} f_{\mathrm{x}}(x, y) +f_{\mathrm{y}}(x, y),\label{cost_cen}
\end{align}

We now present the following proposition, demonstrating that the centralized CP shaving model is equivalent to a travail centralized peak shaving model. This model minimizes the total peak demand of both agents while also incorporating peak demand charges and shifting costs.
\begin{proposition}\label{centralizd_proposition}
    \emph{Centralized peak shaving model}.
    The centralized CP shaving model (\ref{cost_cen}) is equivalent to the trivial concave peak shaving model as follows:
        \begin{align}
            (x^*,y^*) \in \arg\max_{x,y}\ -\pi \max \{S_{1}(x,y),S_2(x,y)\} - \alpha_{\mathrm{x}} x^2 - \alpha_{\mathrm{y}}y^2\label{convex_centralized_mode}
        \end{align}
\end{proposition}

\begin{proof}
    Take $f_\mathrm{x},f_{\mathrm{y}}$ as defined in (\ref{agenta}) into (\ref{cost_cen}), we have
        \begin{align}
         f_\mathrm{x}(x, y) &+f_{\mathrm{y}}(x, y) =    -\pi(X_{1}+x) I(x,y) 
        -\pi(X_{2}-x) (1-I(x,y))
        \nonumber\\
        &-\pi(Y_{1}+y) I(x,y) 
        -\pi(Y_{2}-y) (1-I(x,y))
        - \alpha_{\mathrm{y}} y^2 - \alpha_{\mathrm{x}} x^2 \nonumber\\
        &= -\pi S_{1}(x,y)I(x,y)
        -\pi S_{2}(x,y) (1-I(x,y))
        - \alpha_{\mathrm{y}} y^2 - \alpha_{\mathrm{x}} x^2\nonumber\\
        &=-\pi \max\{S_1(x,y),S_2(x,y)\} - \alpha_{\mathrm{y}} y^2 - \alpha_{\mathrm{x}} x^2.
        \end{align}
    Note that the peak shaving model in (\ref{convex_centralized_mode}) is concave because maximize two concave (linear) functions $S_1(x,y),S_2(x,y)$ is concave; thus, we prove the Proposition.
\end{proof}



        

This model assumes direct control over each agent's demand, achieving peak shaving at minimal cost. However, while its performance is strong, the assumption of direct control over each agents' demand is unrealistic in practice, as agents (customers) typically have independent strategies and privacy concerns. Therefore, we consider \eqref{convex_centralized_mode} as a benchmark to compare with our CP shaving game model and analyze peak shaving effectiveness and efficiency loss in the following sections.


\subsection{Peak shaving ratio analysis}
We show in the following theorem that the CP shaving game can achieve the same effectiveness in reducing the system peak demand at its equilibrium.
\begin{theorem}\label{peak_shacing}
\emph{Peak shaving performance}.
    The peak shaving effectiveness of the CP shaving game model $G$ at equilibrium is always 1, i.e.,
    \begin{align}
    \frac{\max \{S_1(x^*,y^*),S_2(x^*,y^*)\}}
    {\max \{S_1(x_{\mathrm{cen}}^*,y_{\mathrm{cen}}^*),S_2(x_{\mathrm{cen}}^*,y_{\mathrm{cen}}^*)\}} = 1,
    \end{align}
    for all $\pi, \alpha_\mathrm{x},\alpha_{\mathrm{y}}, X,Y>0$; where $x^*,y^*$ is the game equilibrium results and $x_{\mathrm{cen}}^*, y_{\mathrm{cen}}^*$ is the centralized peak shaving results.
\end{theorem}
\begin{proof}[Sketch of the proof]
    We prove this Theorem by first analytically writing the best solution for the centralized model shown in Proposition \ref{centralizd_proposition}, and from Theorem \ref{n_agent_game}, we have the game model equilibrium. Then take the solutions of the centralized model and the game model under the same condition into the peak shaving effectiveness definition in this Theorem. The detailed proof is provided in the appendix.
\end{proof}

This Theorem shows the game model always reaches the same peak shaving performance when compared with the centralized model. The reason is that agents in both the game model and centralized model shift demand as much as possible to avoid CP charge without burdening more by their shifting penalty. Specifically, under concave game conditions, both agents' shifting capability is constrained by their critical points, preventing them from fully balancing their demand—similar to the centralized model. Under quasiconcave and non-concave game conditions, system demand is balanced across the two periods, also aligning with the centralized model. This suggests that the game model is practically applicable, as utilities and operators can implement the game framework for CP shaving without compromising peak shaving performance. In other words, operators can still achieve their peak shaving targets while using the game model. However, although the overall peak shaving performance is the same, agents' individual demand shifting is different due to the information barrier, which reflects as cost of reaching the peak shaving performance. We then analyze the cost by showing the efficiency loss in the following section. 


\subsection{Efficiency loss analysis}
In this section, we study the efficiency loss affected by gaming agents' strategic behavior in all three game structures, and we provide the following main results.
\begin{theorem}\label{PoA_agent_equity}
    \emph{Efficiency loss with agent equity.} 
    Given the efficiency loss defined as  
    \begin{subequations}
    \begin{align}
        P  =  \frac{f_{\mathrm{x}}(x^*,y^*) + f_{\mathrm{y}}(x^*,y^*)}{f_{\mathrm{x}}(x_{\mathrm{cen}}^*,y_{\mathrm{cen}}^*) +f_{\mathrm{y}}(x_{\mathrm{cen}}^*, y_{\mathrm{cen}}^*)},\label{poa}
    \end{align}
    under the quasiconcave and non-concave game condition as described in Proposition \ref{game_proper}, the efficiency loss increases with the disparities among agents, as measured by the marginal shifting cost $\alpha_{\mathrm{x}}x^*,\alpha_{\mathrm{y}}y^*$, i.e.,
    \begin{align}
        \frac{\partial P } {\partial [(\alpha_{\mathrm{x}}x^*-\alpha_{\mathrm{y}}y^*)^2]} > 0
    \end{align}
     where $x^*,y^*$ is the game equilibrium result and $x_{\mathrm{cen}}^*,y_{\mathrm{cen}}^*$ is the centralized peak shaving results.          
    \end{subequations}
\end{theorem}

\begin{proof}[Sketch of the proof] 
    From Theorem \ref{n_agent_game}, we know the NE of the game model under quasiconcave and non-concave game conditions. Combined with the centralized model solution obtained from Theorem \ref{peak_shacing}, we show the difference of nominator and denominator of the efficiency loss as defined in (\ref{poa}) can be expressed as a Euclidean distance between agents' marginal shifting cost, i.e., $(\alpha_{\mathrm{x}}x^*-\alpha_{\mathrm{y}}y^*)^2$. The detail is provided in the appendix.
 \end{proof}

Under non-concave and quasiconcave game conditions, although the overall CP charge is always the same as the centralized model due to the same peak shaving effectiveness, the shifting cost $\alpha_{\mathrm{x}}x^{*2}+\alpha_{\mathrm{y}}y^{*2}$ increases due to the information barrier. We show in this theorem that the efficiency loss increases with agents' disparities, quantified by their marginal shifting cost. This insight provides a pathway for designing mechanisms that balance agents' marginal shifting costs, which also enhances fairness while improving overall system effectiveness.
For instance, agents with both a large shifting penalty parameter and a high shifting amount bear a higher marginal shifting cost. Here, the penalty parameter reflects greater comfort loss from demand shifting, while the high shifting amount arises from significant demand differences requiring adjustment due to CP charges. In such cases, operators or utilities could regulate these agents to shift less demand or provide incentives to reduce their penalty parameters. 
As some policies already suggest subsidizing disadvantaged customers, this analysis offers guidance on how to allocate subsidies effectively among customers~\citep{PGE_policy,NY_policy}. By balancing marginal shifting costs, this approach simultaneously improves system fainness and efficiency.

 
\begin{theorem}\label{PoA_CPgame_type}
    \emph{Efficiency loss with CP shaving game type.} 
    For given $\pi, S, \alpha_\mathrm{x},\alpha_\mathrm{y}>0$, the efficiency loss defined as Theorem \ref{PoA_agent_equity}
    is highest at equilibrium for quasiconcave games, followed by non-concave games, and lowest for concave games, where it is always equal to 1, i.e.,
    \begin{align}
        P(\text{Quasiconcave game}) \geq P(\text{Non-concave game})  \geq P(\text{Concave game}) =1.\nonumber
    \end{align}
\end{theorem}

\begin{proof}[Sketch of the proof] 
    We first show $P = 1$ is always true for the concave game due to the same solution structure between the centralized model and concave game model based on the results from Theorems \ref{n_agent_game} and \ref{peak_shacing}. Then from Theorem \ref{PoA_agent_equity}, we know the efficiency loss expression as Euclidean distance between agents' marginal shifting cost under quasiconcave and non-concave game conditions, indicating $P \geq 1$. Thus, quasiconcave games and non-concave games always cause higher (or equal) efficiency loss than concave games. We then show by fixing $\pi, S, \alpha_\mathrm{x},\alpha_\mathrm{y}>0$, the agent's baseline demand $X_{1},X_{2},Y_1,Y_2$ can vary to cause different game structures, so as to different solution structure $x^*,y^*$. Combined with the fact that quasiconcave and non-concave games always balance system demand, i.e., $x^* +y^* = b$ from Theorem \ref{n_agent_game}, we prove the quasiconcave game solution shows a higher difference between $x^*, y^*$, indicating higher efficiency loss than a non-concave game. The detailed proof is provided in the appendix. 
 \end{proof}

This theorem demonstrates the impact of the CP shaving game type on the efficiency loss, where the game type connects to the agents' flexibility (shifting capability). We first fix the system conditions $\pi, S$ while allowing the agents' conditions to vary. From Theorem \ref{n_agent_game}, we know that agents' flexibility is shaped by the game structure, which in turn is influenced by the shifting penalty parameters $\alpha_\mathrm{x},\alpha_\mathrm{y}$ and baseline demand conditions $X_{1},X_{2},Y_1,Y_2$. To control the influence these parameters, we also fix $\alpha_\mathrm{x},\alpha_\mathrm{y}$, making the efficiency loss dependent only on $x^*,y^*$, which reflects the difference in the game structure. As the balance point (demand difference) $b_{\mathrm{x}}=(X_{2}-X_{1})/2$ increases, agent x's flexibility decreases, and the game type transitions from a quasiconcave game to a non-concave game, and finally to a concave game. This follows from the conditions 
$\{0\leq b_{\mathrm{x}} \leq r_{\mathrm{x}}\} 
\cap \{-r_{\mathrm{y}} \leq b_{\mathrm{y}} \leq r_{\mathrm{y}}\}$ 
to $\{0\leq b\leq r_{\mathrm{x}} + r_{\mathrm{y}} \} 
\cap \{b_\mathrm{x} > r_{\mathrm{x}} \cup b_\mathrm{y} < -r_{\mathrm{y}} \cup b_\mathrm{y} > r_{\mathrm{y}}\}$ 
to $\{r_{\mathrm{x}}+r_{\mathrm{y}}<b\}$ as stated in Proposition \ref{game_proper}.
Thus, the efficiency loss increases with agents' flexibility and is also reflected in the game type change. This highlights the role of inflexible agents in overall system performance, despite utilities and operators generally favoring flexible agents for their ability to enhance grid stability, integrate renewable energy, and respond to pricing signals~\citep{DR}. 

It is also evident that fixing $X_{1},X_{2},Y_1,Y_2$ while allowing $\alpha_{\mathrm{x}},\alpha_{\mathrm{y}}$ to vary yields the same results. Note that this theorem also shows that under the concave game condition, the system's performance is always equivalent to that of the centralized model, indicating no harm to utility companies and agents when applying the concave CP shaving game.

\section{Extension to Multi-agent Games}\label{multi_section}
In this section, we consider a generalization of our model from the two-agent two-period CP shaving game to the multi-agent two-period CP shaving game, and we analyze the NE, stability/convergence, as well as the customers' strategic behavior. 

\subsection{Multi-agent CP shaving game}
We denote $X_{i,1},X_{i,2}$ as agent $i$'s baseline demand in periods 1 and 2, respectively,  and extend the system baseline demand $S_{\mathrm{b},1},S_{\mathrm{b},2}$ from the sum of two agents' demands to the sum over $N$ agents at each period. Without loss of generality, We maintain the assumption that the baseline demand in period 2 is greater than in period 1, i.e., $S_{\mathrm{b},1} < S_{\mathrm{b},2}$. We use $\alpha_i$ to denote agent $i$'s penalty parameter, and $x_i$ as its demand-shifting strategy. Note that we do not impose any assumptions on individual agents' baseline demands, ensuring that all $x_i$ remain unrestricted. We continue to use $S_1,S_2$ to represent the total system demand, incorporating all agents' demand shifts. The multi-agent game model is then formally defined as $G'=(N,\mathcal{X},U)$, where
\begin{itemize}
    \item $N=\{1,2,...,|N|\}$ with index $i$ is the agent set, where $|\cdot|$ represents the number of elements in a set (otherwise, it denotes the absolute value). Specifically, we use $-i$ to refer to all agents except agent $i$.
    \item $\mathcal{X}=\times_{i\in N}\mathcal{X}_{i}$ is the strategy set formed by the product topology of each agent's individual strategy set $\mathcal{X}_{i}$. 
    \item $U =\{f_i|i\in N\}$ is the payoff function set, where $f_i: \mathcal{X} \rightarrow \mathbb{R}$.
\end{itemize}  
Then agent $i$'s payoff function as described in (\ref{agenta}) can be generalized as follows:
\begin{subequations}\label{multi_agent_game}
    \begin{align}
         \max_{x_i} f_i (x_i,x_{-i}) &= -\pi(X_{i,1}+x_i) I(x_i,x_{-i})-\pi(X_{i,2}-x_i) (1-I(x_i,x_{-i}))
        - \alpha_{i} x_i^2, \\
        I(x_i,x_{-i}) &= \begin{cases} 
      1 & S_1(x_i,x_{-i}) - S_2(x_i,x_{-i})\geq 0 \\
      0 & S_1(x_i,x_{-i}) - S_2(x_i,x_{-i}) < 0 
        \end{cases},\label{step_function} \\
        S_1(x_i,x_{-i}) &= \sum_{i\in N}(X_{i,1}+x_i) = S_{\mathrm{b},1}+\sum_{i\in N}x_i, \\
        S_2(x_i,x_{-i}) &= \sum_{i\in N}(X_{i,2}-x_i) = S_{\mathrm{b},2}-\sum_{i\in N}x_i ,
    \end{align}
\end{subequations}

Accordingly, the system average demand $S$ and the balance point $b$ as defined in Definition \ref{Special_points} are extended from the two-agent case to the $N$ agent case. Thus, we can conclude that the multi-agent CP shaving game $G'$ retains the properties described in Proposition \ref{game_proper}, meaning that the CP shaving game $G'$ is 
\begin{subequations}\label{multiagent_condition}
\begin{enumerate}
    \item Concave if 
    \begin{align}
        b = \sum_{i\in N}b_i> \sum_{i\in N}r_{i} > 0;\label{concaveG'}
    \end{align}
    \item Quasiconcave if
    \begin{align}
        -r_{i} \leq b_{i} \leq r_{i},\forall i\in N;\label{quasiconcaveG'}
    \end{align}
    \item Non-concave if 
    \begin{align}
        \{0 \leq b \leq \sum_{i\in N}r_{i}\} \cap \{\exists i \in N, b_{i}<-r_{i} \cup \exists i \in N, b_{i}>r_{i} \}.\label{non_concaveG'}
    \end{align}
\end{enumerate}
\end{subequations}
The concave and quasiconcave conditions are intuitive and can be derived following the same process outlined in Proposition \ref{game_proper}. The non-concave condition, on the other hand, is determined by the complementary set of the concave and quasiconcave conditions with respect to $\mathbb{R}$.

In terms of NE solutions, the definition of NE should be slightly modified from Definition \ref{Nash_Equ_defi} to account for all agents. Specifically, $(x^*,x_{-i}^*)$ is a pure strategy NE of the game $G'$ if and only if 
\begin{align}
    f_i(x_i^*,x_{-i}^*) \geq f_{i}(x_i,x_{-i}^*), i\in N, x_i\in \mathcal{X}_i,x_{-i}\in \mathcal{X}_{-i}.\label{NE_multi_agent}
\end{align}

Unlike the two-agent setting, analyzing each agent's capability case by case becomes infeasible in the multi-agent setting. Therefore, we analyze the NE separately for each game type. First, for the concave game, we establish the existence of a unique NE following~\citep{rosen1965}. The NE is reached when all agents shift demand up to their critical points, as defined in Definition \ref{Special_points}:
\begin{align}
    x_i^* &= r_{i}, i\in N.\label{multi_concave_solution}
\end{align}
Noted that when all agents adopt $x_i^*=r_{i}$, the CP period remains period 2 and $S_{\mathrm{b},1} < S_{\mathrm{b},2}$. In this case, all agents shift demand from period 2 to period 1 until reaching their respective critical points.

The NE of the quasiconcave game is less straightforward to derive. To analyze it rigorously, we introduce the following proposition.
\begin{proposition}\emph{Existence and uniqueness of NE in multi-agent quasiconcave CP shaving game.} \label{uniques_exist_theorem}
    The quasiconcave CP shaving game $G'$ as described in (\ref{multi_agent_game}) and satisfy (\ref{quasiconcaveG'}) has a unique pure-strategy NE $(x_i^*,x_{-i}^*)$ as defined in (\ref{NE_multi_agent}), where
        \begin{align}
         &x_i^* = b_{i}, i\in N, 
        S_1(x_i^*,x_{-i}^*) = S_{\mathrm{b},1}+\sum_{i\in N}x_i^* = S_{\mathrm{b},2} - \sum_{i\in N}x_i^* = S_2(x_i^*,x_{-i}^*).\label{unique_NE_condition}
        \end{align}
\end{proposition}

\begin{proof}[Sketch of the proof]
    To prove this proposition, we establish two supporting lemmas: one for the existence of NE (Lemma \ref{existence}) and another for its uniqueness (Lemma \ref{uniqueness}). The proof of existence follows the framework of~\citep{reny1999}, which ensures existence by demonstrating that the sum of agents' payoff functions is upper semi-continuous in $x_i\in \mathcal{X}$ and game $G'$ is payoff security, which we formally define in the proof. For uniqueness, we build upon the results of~\citep{rosen1965}. Specifically, we first establish the uniqueness of NE in both periods 1 and 2. Then, leveraging the two-agent NE characterization from Lemma \ref{lemma_NE2agent}, we extend these results to the multi-agent setting by iteratively partitioning the agents into two groups and applying the two-agent results. The full proof is provided in the appendix.
\end{proof}

This proposition establishes that the quasiconcave CP shaving game $G'$ has a unique NE, which always occurs at the agents' balance points $b_i$. Notably, under both concave and quasiconcave conditions, the NE results align with those in the two-agent setting, as the extension from two agents to multiple agents follows a straightforward generalization. This suggests an approach where the NE results of the two-agent non-concave CP shaving game $G$ serve as a foundation for analyzing the $N$-agent non-concave CP shaving game $G'$, whose NE is evidently more complex. Fundamentally, Theorem \ref{n_agent_game} highlights a key insight: system demand is always balanced across the two periods, i.e., $x^* + y^* = b$. In this process, the less flexible agent—determined by shifting penalty parameters and baseline demand—reaches its critical point first, while the more flexible agent continues adjusting its demand to balance system demand over the two periods.

We define agents whose baseline peak demand coincides with the system baseline CP period as \emph{CP-period agents}, while those whose baseline peak demand does not align with the system baseline CP period are referred to as \emph{non-CP-period agents}. The set containing all CP-period agents is denoted as $N_{\mathrm{cp}}$, and the set of all non-CP-period agents as $N_{\mathrm{ncp}}$, such that $N_{\mathrm{cp}}\cup N_{\mathrm{ncp}} = N$ and $N_{\mathrm{cp}}\cap N_{\mathrm{ncp}} =\emptyset$. We now present the following proposition to analyze the equilibrium conditions of the non-concave multi-agent CP shaving game $G'$.


\begin{proposition}\emph{NE in the non-concave multi-agent CP shaving game.}\label{N_agent_extropolate}
        Given non-concave multi-agent CP shaving game $G'$, under the condition of 
        \begin{subequations}
        \begin{align}
        \{\sum_{i\in N_\mathrm{cp}} \min\{r_{i},b_{i}\}  < \sum_{i\in N_\mathrm{cp}}b_{i}\} \cap \{-\sum_{i\in N}r_i \leq b 
        \leq \sum_{i\in N}r_i\};,\label{multi_agent_non_concave_condition}
        \end{align}   
        the hybrid NE solution, where for each CP-period agent and for the whole non-CP-period agent set, is given by
        \begin{align}
        x_i^* = \min\{r_i,b_i\},i\in N_{\mathrm{cp}},\sum_{i\in N_{\mathrm{ncp}}}x_i^* = b-\sum_{i\in N_{\mathrm{cp}}}\min\{r_i,b_i\}.\label{multi_agent_non_concave1}
         \end{align}                 
        Otherwise, the hybrid NE solution, where for each non-CP-period agent and the overall CP-period agent set, is given by
        \begin{align}
        x_i^* = \max\{-r_i,b_i\},i\in N_{\mathrm{ncp}}, \sum_{i\in N_{\mathrm{cp}}}x_i^* = b-\sum_{i\in N_{\mathrm{ncp}}}\max\{-r_i,b_i\},\label{multi_agent_non_concave2}
         \end{align}  
        and the corresponding condition is
        \begin{align}
            \{\sum_{i\in N_\mathrm{ncp}} \max\{-r_i,b_i\}  > \sum_{i\in N_\mathrm{ncp}}b_i\} \cap  \{-\sum_{i\in N}r_i \leq b 
        \leq \sum_{i\in N}r_i \};
        \end{align}
        \end{subequations}
\end{proposition}

\begin{proof}[Sketch of the proof]
    The proof of this proposition is based on introducing two virtual agents, each representing the aggregated baseline demand and shifting penalty conditions of the CP-period and non-CP-period agent sets. While these virtual agents share the same baseline conditions as their respective agent sets, their strategy structure differs, as the agent sets' strategies result from the aggregation of individual agents' decisions. Using Theorem \ref{n_agent_game}, we derive the best strategies for the virtual agents. From the best response analysis, we establish that both the agent sets and the virtual agents exhibit the same strategic behavior due to their identical baseline conditions. Specifically, both will shift demand away from or toward the CP period to minimize costs, with cost reduction directly related to shifting amounts. Finally, by linking the virtual agents' strategies to those of the agent sets, we derive the equilibrium conditions for the CP-period and non-CP-period agent sets. The detailed proof is provided in the appendix.
\end{proof}

This proposition establishes that the equilibrium solution can be analytically determined for one group of agents, either the CP-period agent set $N_{\mathrm{cp}}$ or the non-CP-period agent set $N_{\mathrm{ncp}}$. However, for the remaining set, only the aggregate performance can be determined, rather than individual agent strategies. As a result, we refer to this as a hybrid NE, where one subset of agents follows a strict NE while the other retains internal flexibility, adjusting their individual strategies while maintaining the overall balance at the set level. In practice, more flexible agents, those with lower shifting penalties and greater baseline demand differences, will contribute more to the overall shifting performance of their set. A key insight from this result is that system demand remains balanced under non-concave game conditions, consistent with the findings from the two-agent setting. This conclusion provides a foundation for evaluating system efficiency losses and peak shaving effectiveness.

We then analyze the stability of the equilibrium point in the multi-agent CP shaving game $G'$ by building on the two-agent game results presented in Theorem \ref{dynamics_theorem}. Extending the gradient vector to incorporate all agents' shifting decision variables $x_i,i\in N$, we still denote the system dynamics as (\ref{switched_system}). we continue to represent the system dynamics using (\ref{switched_system}). Following the same proof structure, we derive multiple Lyapunov functions for periods 1 and 2, demonstrating that the globally uniform asymptotically stability also holds for the multi-agent CP shaving game $G'$. Under concave and quasiconcave conditions, the unique NE described in (\ref{multi_concave_solution}) and Proposition \ref{uniques_exist_theorem} serve as the stable equilibrium point in the system dynamics. In the non-concave case, the system dynamics in periods 1 and 2 converge to a switching surface where the payoffs of both periods are equal, ensuring system demand remains balanced, as characterized by the hybrid NE in Proposition \ref{N_agent_extropolate}.
The strategy set $\mathcal{X}_{\mathrm{s}}$ is then extended to incorporate all agents' decisions as follows:
\begin{align}
    \mathcal{X}_{\mathrm{s}} = \{x_i|\sum_{i\in N}\alpha_{i} x_i^2 + \pi(S_{\mathrm{b},1}+\sum_{i\in N}x_i) >0 
    \cup \sum_{i\in N}\alpha_{i} x_i^2 
    + \pi(S_{\mathrm{b},2} - \sum_{i\in N}x_i) >0 \}. \label{multi_agent_strategy_set}
\end{align}

With this stability property, we further show that a finite learning rate can be selected such that a gradient-based algorithm using (\ref{gradient_algorithm}) as an updating rule converges to the equilibrium point in the multi-agent CP shaving game $G'$. The key distinction in the multi-agent setting is that convergence shifts from each agent’s individual payoff across the two periods to the aggregated payoff of all CP-period agents and non-CP-period agents. This also confirms that under non-concave conditions, agents can internally adjust their demand-shifting strategies while maintaining overall system demand balance across the two periods.
\begin{remark}\emph{Determination of equilibrium point in the multi-agent CP shaving game.}\label{global_convergence_N_agent}
    A finite learning rate vector $\tau_{\mathrm{cp},h}\in \mathbb{R}^{N_{\mathrm{cp}}}$ for CP-period agents in the finite difference approximation to the system dynamics in the multi-agent CP shaving game can be selected such that when $\sum_{i\in N_{\mathrm{cp}}}X_{i,1} +x_i \geq \sum_{i\in N_{\mathrm{cp}}} X_{i,2} - x_i$,
    \begin{subequations}
    \begin{align}
        \sum_{i\in N_{\mathrm{cp}}}-f_{i,1}(x_{i,\overline{h}}) < \sum_{i\in N_{\mathrm{cp}}} -f_{i,1}(x_{i,\underline{h}}),
        \sum_{i\in N_{\mathrm{cp}}} -f_{i,2}(x_{i,\overline{h}}) > \sum_{i\in N_{\mathrm{cp}}} -f_{i,2}(x_{i,\underline{h}}). 
    \end{align}
     When $\sum_{i\in N_{\mathrm{cp}}}X_{i,1} +x_i < \sum_{i\in N_{\mathrm{cp}}} X_{i,2} - x_i$,   
        \begin{align}
        \sum_{i\in N_{\mathrm{cp}}}-f_{i,1}(x_{i,\overline{h}}) > \sum_{i\in N_{\mathrm{cp}}} -f_{i,1}(x_{i,\underline{h}}),
        \sum_{i\in N_{\mathrm{cp}}} -f_{i,2}(x_{i,\overline{h}}) < \sum_{i\in N_{\mathrm{cp}}} -f_{i,2}(x_{i,\underline{h}}). 
    \end{align}          
    \end{subequations}
    The same also holds for the learning rate vector $\tau_{\mathrm{ncp},h}\in \mathbb{R}^{N_\mathrm{ncp}}$ for non-CP-period agents.
\end{remark}

To conclude, the game framework remains effective in the multi-agent, two-period setting, as the (hybrid) NE continues to exist and exhibits globally uniform asymptotically stable within the strategy set defined in (\ref{multi_agent_strategy_set}). Additionally, a gradient-based algorithm, using an updating rule that approximates the system dynamics of the multi-agent CP shaving game, can reliably compute the equilibrium point. Building on this foundation, we next analyze the impact of customers' strategic behavior in the multi-agent setting.

\subsection{Customers strategic behavior in multi-agent CP shaving game}
In this section, we analyze the peak shaving effectiveness and efficiency loss in the multi-agent setting. The centralized model is formulated as $(x_{i,\mathrm{cen}}^*, x_{-i,\mathrm{cen}}^*)\in \arg\max_{x_i,x_{-i}} \sum_{i\in N} f_i(x_i,x_{-i})$, representing the optimal solution that maximizes the total agent payoffs. Combined with Proposition \ref{centralizd_proposition}, under the concave game condition, the solution to the concave centralized model can be directly obtained using the first-order optimality conditions.

The peak shaving effectiveness, as defined in Theorem \ref{peak_shacing}, extends from the two-agent setting to the multi-agent setting by replacing agent y to $-i$, denoting all agents except $i$. It remains equal to 1 at the equilibrium of the game model because system demand is always balanced under all game conditions, just as in the centralized model.
The efficiency loss in the multi-agent setting is given by $P_N = \sum_{i\in N}f_{i}(x_i^*,x_{-i}^*)/\sum_{i\in N}f_{i}(x_{\mathrm{cen},i}^*,x_{\mathrm{cen},-i}^*)$. Since the centralized model is concave, derived from Proposition \ref{centralizd_proposition}, and $x_{\mathrm{cen},i}^*,x_{\mathrm{cen},-i}^*$ is its unique minimizer, we always have $P\geq 1$. Specifically, under the concave game condition, the game model is equivalent to the centralized model, ensuring that $P=1$. 

As the number of agents increases, the efficiency loss under quasiconcave and non-concave games is influenced by changes in the game structure. This leads to the following remark.
\begin{remark}\label{PoA_multi_agent}
    \emph{Game type with agent numbers.}
    As $N$ increases, the game structure will more likely be a non-concave game. 
    \end{remark}
It is intuitive that the non-concave game type becomes more prevalent as the number of agents increases because the probability that all agents are either capable or non-capable decreases exponentially. As analyzed in Theorem \ref{PoA_CPgame_type}, agents' flexibility influences the game structure and, consequently, the efficiency loss. This implies that smaller systems are more sensitive to variations in agent flexibility, whereas larger systems can mitigate the impact of highly flexible agents on overall efficiency. 

This insight provides guidance for operators and utilities implementing game-based CP shaving programs. Specifically, selecting more flexible customers to form a larger CP shaving system while grouping a few inflexible agents into a smaller CP shaving system can enhance efficiency. These two systems would operate independently, each determining its own CP periods based on the participating customers. This targeted structuring can improve overall system performance and mitigate efficiency loss.

\section{Numerical Example}
In this section, we use numerical simulations to show the CP shaving game solution, and we show that the numerical test aligns with our theoretical analysis. We set the CP charge price to $\pi=1$ and set three two-agent two-period cases and one multi-agent two-period case to show the CP shaving game performance. In the multi-agent setting, we also test the influence of the number of agents on efficiency loss.

\subsection{Two-agent two-period CP shaving game}
We set the baseline demand as $X_{1} = 3, X_{2} = 10,Y_{1} = 6,Y_{2} = 3$, then change the shifting penalty parameters to change the agent capability.

(1) Set $\alpha_{\mathrm{x}} = 0.1, \alpha_{\mathrm{y}}=0.2$, then all agents' are capable according to Definition \ref{capably} and the game is quasiconcave. We show the solution path in Fig. \ref{quasiconcave_case}, where the system alternates between periods 1 and 2. With each transition back to the previous period, the cost either decreases or increases in a way that reduces the cost difference for each agent between the two periods. Finally, the cost distance reduces to zero for each agent as they balance their demand over the two periods, and the system converges to the equilibrium point, corresponding to (\ref{fully_compet1}). Compared with the solution from the centralized model, we observe that there are huge shifting changes and an increase in the system's overall cost due to anarchy, reflected as the efficiency loss is 1.125, indicating anarchy increases the system cost by 12.5\% compared to the centralized method. The peak shaving ratio is the same because they all balance the system demand.
\begin{figure}
    \centerline{\includegraphics[width=0.8\textwidth]{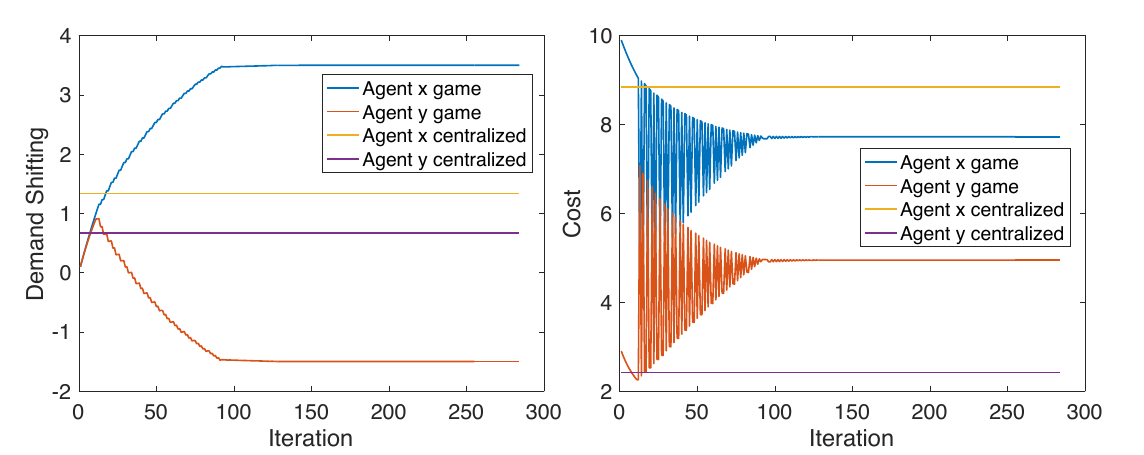}}
    \caption{Convergence performance of two-agent quasiconcave CP shaving game.}
    \label{quasiconcave_case}
\end{figure}

(2) Set $\alpha_{\mathrm{x}} = 0.1, \alpha_{\mathrm{y}}=0.5$, then agent y is non-capable, and agent x is capable according to Definition \ref{capably}, and the game is non-concave. We show the solution trajectory in Fig. \ref{non_concave_case}, which is similar to the quasiconcave condition. However, each agent's cost difference over the two periods does not reach zero upon convergence, as neither agent fully self-balances its demand. Instead, the more flexible agent (agent x) achieves greater cost reduction by lowering its demand during the CP period. The results align with (\ref{mix2}). Compared to the quasiconcave game, we observe that the demand shifting resulting from the game model more closely approximates the centralized solution, and the increase in system cost due to anarchy is reduced, with an efficiency loss of 1.0941.
\begin{figure}
    \centerline{\includegraphics[width=0.8\textwidth]{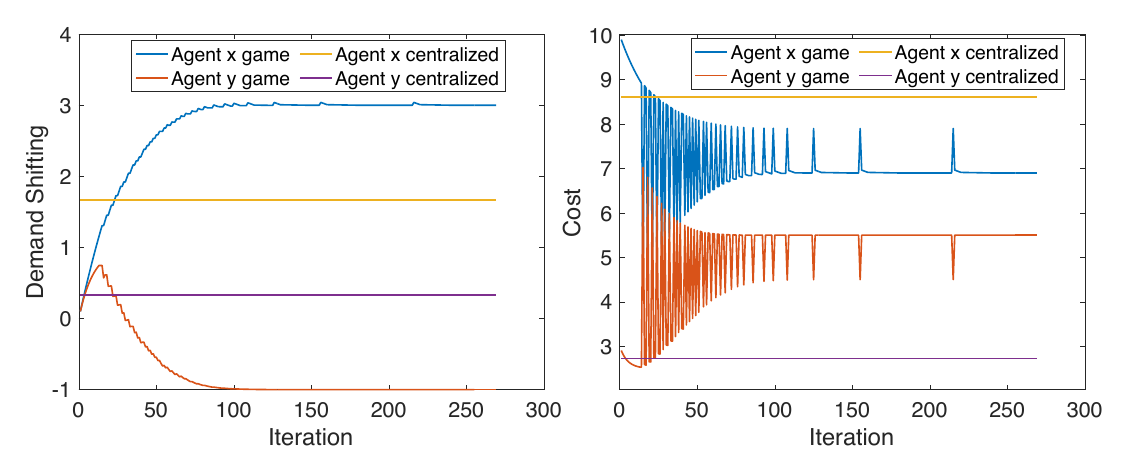}}
    \caption{Convergence of two-agent non-concave CP shaving game.}
    \label{non_concave_case}
\end{figure}

(3) Set $\alpha_{\mathrm{x}} = 0.6, \alpha_{\mathrm{y}}=0.5$, then both agents can't change the CP period together, and the game is concave, where only period 2 is active during the entire solution trajectory. Thus, as we show in Fig. \ref{concave_case}, the trajectory doesn't jump and gradually reduces to the critical point, which corresponds to (\ref{fully_coopera2}). Obviously, the demand shifting converges to the same point as the centralized model, and the efficiency loss is 1, indicating these two models are exactly equivalent.
\begin{figure}
    \centerline{\includegraphics[width=0.8\textwidth]{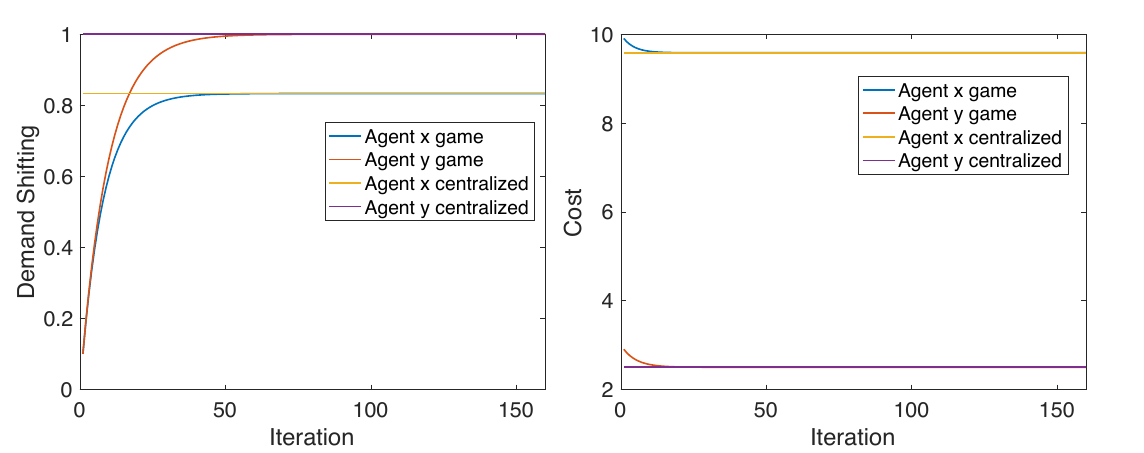}}
    \caption{Convergence performance of two-agent concave CP shaving game.}
    \label{concave_case}
\end{figure}

\subsection{Multiple agents and two-period CP shaving game}
We create a six-agent system with baseline demand and shifting penalty parameters, as Table \ref{talbe11} shows. We first notice the game is non-concave because all agents can change CP period together, but not all agents are capable, i.e., agents 3 and 4 are non-capable. Also, we know agents 1, 3, and 6 are non-CP-period agents, and agents 2, 4, and 5 are CP-period agents.

Fig. \ref{multiagent_converge}(a) shows the convergence and cost trajectory, and Fig. \ref{multiagent_converge}(b) shows the cost trajectory of the non-CP period agent set and the CP period agent set. Noted that the results align with our analysis in Proposition \ref{N_agent_extropolate}. Specifically, non-CP period agents shift demand either to the balance point or until they reach their critical limit. For example, agents 1 and 6 balance their demand across the two periods, while agent 3 reaches its shifting limit. After approximately 120 iterations, system demand balances over two periods (Fig. \ref{multiagent_converge}(b)), and CP-period agents begin internally adjusting their demand. Note that agents 4 and 5, which have higher shifting penalty parameters than agent 2, gradually reduce their demand shifting, while agent 2 increases its shifting. From the cost trajectory (Fig. \ref{multiagent_converge}(a)), agent 4, with the highest shifting penalty, reduces its demand shifting more significantly than agent 5.
Compared to the centralized model, the CP shaving game equilibrium deviates significantly, with all agents shifting more demand due to the information barrier. This increases the cost and results in an efficiency loss of 1.1317 while maintaining the same peak shaving performance.
\begin{table}[]
\centering
\caption{Agents' parameters and solution of the game and centralized model}
\begin{tabular}{@{}cccccccc@{}}
\toprule
Agent   & 1   & 2   & 3   & 4   & 5   & 6   & Total           \\ \midrule
Baseline demand 1  & 7   & 3   & 10  & 1   & 2   & 5   & 28               \\
Baseline demand 2  & 3   & 13  & 4   & 4   & 6   & 3   & 33               \\
Penalty parameter & 0.2 & 0.1 & 0.4 & 0.5 & 0.2 & 0.1 & \textbackslash{} \\ 
Centralized shifting & 0.36 & 0.72 & 0.18 & 0.15 & 0.36 & 0.73 & 2.5 \\
Game shifting & -2 & 3.85 & -1.25 & 0.93 & 1.97 & -1 & 2.5 \\
Efficiency loss & \multicolumn{7}{c}{1.1317}     \\ 
\bottomrule
\end{tabular}
\label{talbe11}
\end{table}

\begin{figure}
\centering
\subfigure[Individual agent]{\includegraphics[width=0.64\textwidth]{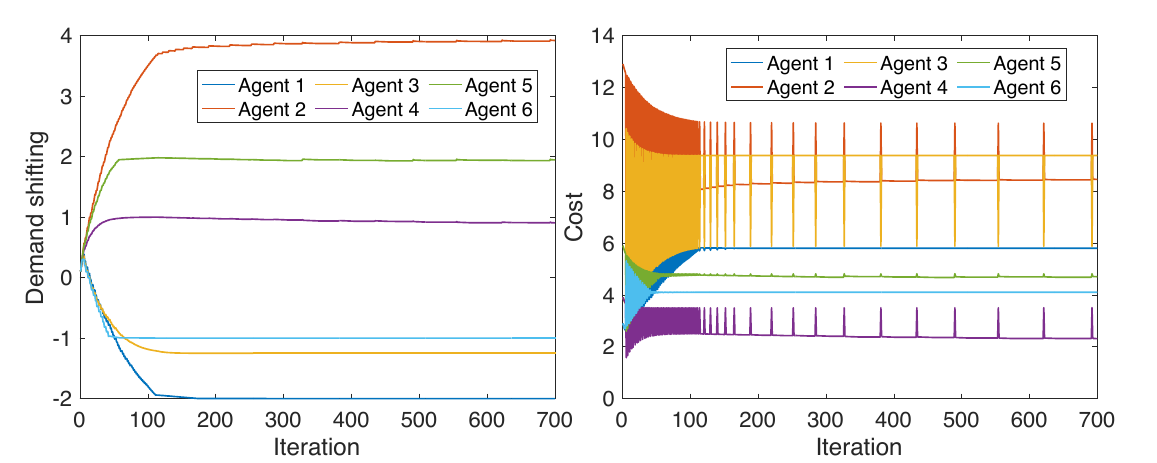}}
\subfigure[CP period and non-CP period agent]
{\includegraphics[width=0.34\textwidth]{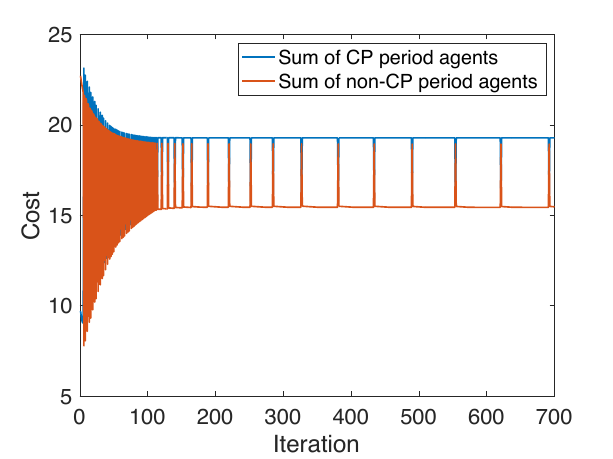}}
\caption{Convergence performance of multi-agent non-concave CP shaving game.}
\label{multiagent_converge}
\end{figure}


\subsection{Agent number impacts on efficiency loss}
In this section, we randomly generate agent samples for systems with varying agent numbers while ensuring that the generated samples satisfy the non-concave and quasiconcave game conditions described in (\ref{quasiconcaveG'}) and (\ref{non_concaveG'}). We allow the baseline system demand consumption to vary freely, thus getting the range $\sum_{i\in N} -r_{i} < b < \sum_{i\in N} r_{i}$. We set the agent's $i,i\in N$ baseline demand as $X_{i,1},X_{i,2}\in(0,15)$ and penalty parameters as $\alpha_{i} \in (0,0.5)$. We loop the agent number from 2 to 50, and each agent number generates 1000 samples to calculate the efficiency loss. We present the results in Fig. \ref{poa_figure}. As our theoretical analysis indicates, efficiency loss is more variable when the agent number is small; also, the game type is more varied with quasiconcave, concave, and non-concave, affected by agents' flexibility, indicating that the small system is more sensitive to the agents' flexibility. When the agent number increases, all game conditions become non-concave, the mean efficiency loss converges, and the variance decreases. This means the large system is more stable and can eliminate the agent's flexibility influence on efficiency loss.


\begin{figure}
    \centerline{\includegraphics[width=0.8\textwidth]{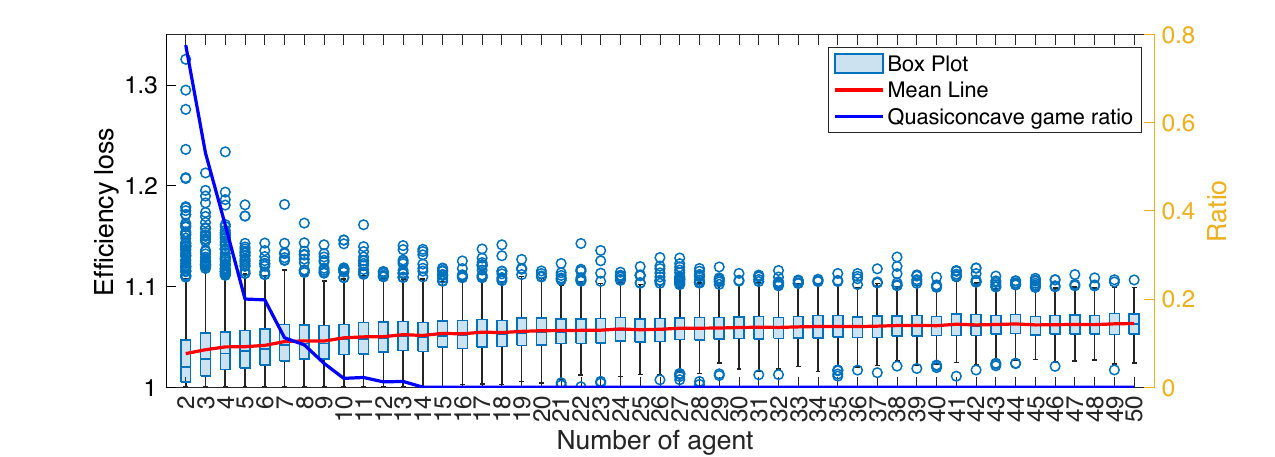}}
    \caption{Agent number impacts on efficiency loss. Each box contains 1000 samples, the circle denotes the outlier, the box upper and lower line denotes the quartiles, and the short blue line inside the box is the median value.}
    \label{poa_figure}
\end{figure}

\subsection{Real-world applications}
In this section, we apply our CP shaving game model to a real-world case in Texas. We utilize data from ERCOT’s 4CP program, under which customers pay a coincident peak charge based on their average demand during one system peak hour in each of the four summer months (June through September). The 4CP charge rate is generally determined by the transmission cost of service (TCOS), which was approximately \$66.76 per kilowatt-year in 2024, implying that a 1 kW peak demand results in \$66.76 in annual electricity bill~\citep{CP_rate}. ERCOT publicly provides each participant’s CP demand record along with the exact CP hours~\citep{4CP_record}. In this study, we focus on the year 2024, during which the CP hours occurred on June 30 (17:45), July 1 (17:00), August 20 (17:00), and September 20 (16:00), and the total participants is 136 (excluding 6 entities with zero CP demand). The final CP demand is the average of these four hourly CP demands.

Due to the proprietary nature of detailed consumption data, and given that our primary focus is evaluating CP demand reduction under a game-theoretic framework, we aggregate all hours from June 1 to September 30 outside the four CP intervals into a single non-CP period. We then use the average demand across this aggregated set to represent each participant’s off-peak consumption.
Using ERCOT’s system-wide hourly demand data~\citep{load_profile}, we estimate the average non-CP period demand for each participant. This approach is based on the assumption that large commercial and industrial customers—the primary participants in the 4CP program—have electricity consumption patterns that are largely driven by their operational capacity and schedules, and typically exhibit relatively consistent load shapes throughout the day~\citep{real_case1, real_case2}. Although some temporal fluctuations exist, averaging over a multi-period window helps smooth out short-term volatility. Based on this, we estimate each participant’s average non-CP period demand by applying its demand ratio in the CP period to the system-wide average non-CP period demand, with an added 50\% random variation to capture heterogeneity across customers. 

We then determine each company's shifting cost parameter based on its demand value. Specifically, we calculate a minimum shifting cost parameter $\alpha_i'$ using the period 2 demand, which, under our assumption, is the baseline CP period. This ensures that the customer's demand does not become negative during the game, and is given by the expression $\alpha_i' = \pi / 2X_{i,2}$. To capture demand flexibility and introduce heterogeneity into the simulation, we add a 20\% random fluctuation to the minimum cost parameter.

Fig. \ref{real_case}(a) presents the CP charges based on actual CP demand data and the charges after applying the CP shaving game model. Markers are used to distinguish CP-period and non-CP-period customers. We can see that customers with higher baseline demand tend to bear a larger share of CP charges in the real system. Because their consumption has a greater impact on the system peak, and because they generally face lower shifting costs, these customers tend to benefit more under the game framework. System-wide, the game leads to a significant reduction in CP charges compared to the real CP-influenced demand data—approximately \$651 million. This effect is further illustrated in the convergence behavior shown in Fig. \ref{real_case}(b), where the total cost (i.e., CP charge plus shifting cost) declines for CP-period customers and slightly increases for non-CP-period customers. The shape of the convergence curve also suggests that the game exhibits quasiconcavity, which is a result of the way we specify the shifting cost parameters. The observed reduction in CP charges arises because we assume all customers in the game model strategically participate to minimize their peak-related cost. In practice, not all customers are interested in investing in CP-shaving efforts. Therefore, our results suggest that increasing customer participation has the potential to further drive down the CP demand and improve system efficiency.


We also compare the game results with those from a centralized peak shaving model, and observe that both approaches always achieve the same peak shaving effectiveness. To evaluate efficiency loss, we vary the random variation in the demand ratio across five levels, 25\%, 50\%, 75\%, 100\%, and 125\%, and generate 50 random samples for each level to assess robustness. We show the efficiency loss in Fig.~\ref{real_case}(c), which increases with the magnitude of variation in non-CP demand relative to CP demand. This is because greater variation requires more interactions to reach convergence. Still, the efficiency loss is mainly within 1.05 and demonstrates robust performance. This aligns with our theoretical analysis and can be attributed to the stabilizing effect of a large number of participating agents.


\begin{figure}
\centering
\subfigure[CP demand performance]{\includegraphics[width=0.32\textwidth]{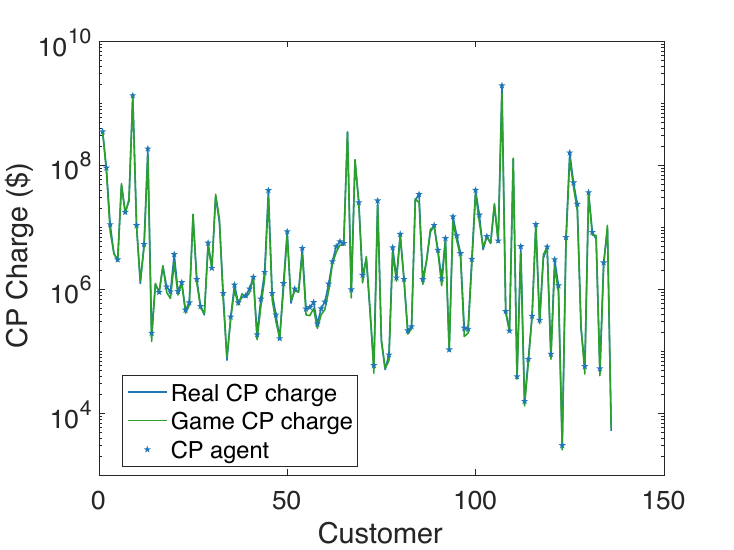}}
\subfigure[Convergence curve of the CP game]
{\includegraphics[width=0.32\textwidth]{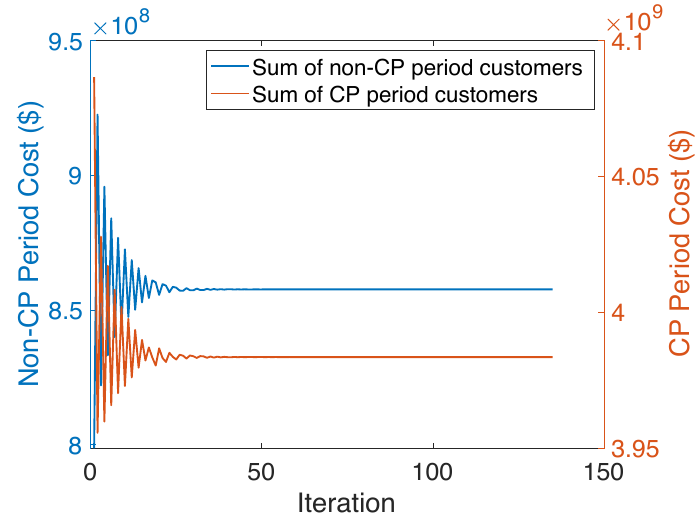}}
\subfigure[Uncertainty impact on demand pattern and shifting cost parameter]{\includegraphics[width=0.32\textwidth]{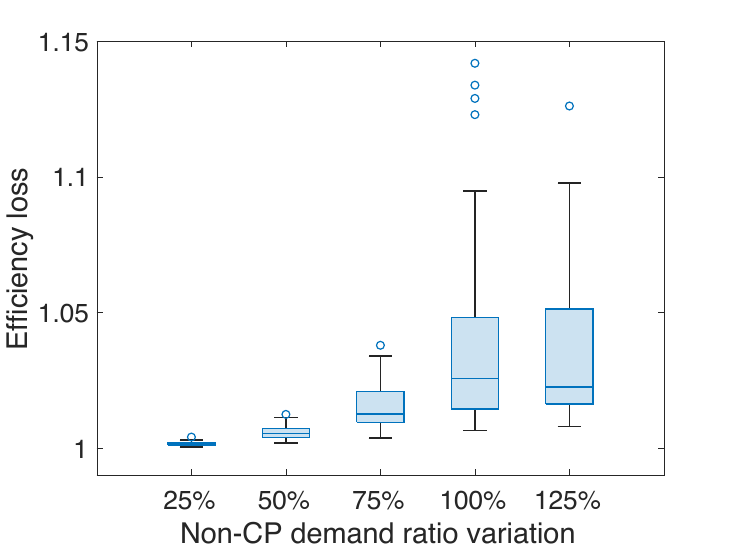}}
\caption{Results with Texas 4CP program.}
\label{real_case}
\end{figure}

\section{Discussion and Conclusion}
Although CP shaving has become a common practice, the interactive influence of agents on CP time remains largely unexplored. This paper takes a first step toward understanding the application of the game framework in the context of CP shaving. We propose a theoretical game framework and show its applicability for the CP shaving problem by analyzing the game structure and analytically deriving the NE for each game structure under a two-agent two-period setting. We prove that the equilibrium points are globally uniformly asymptotically stable, provided that all agents’ demand shifting remains within the baseline limit. We also show that the gradient-based algorithm, with an updating rule that approximates the asymptotically stable process via finite differences, can efficiently compute the equilibrium points. Our analysis reveals that agents' strategic interactions in the game framework achieve the same peak-shaving effectiveness as an equivalent centralized peak-shaving model. However, under quasiconcave and non-concave conditions, the CP shaving game results in efficiency losses. Our findings also persist for a general setting with more than two agents. These theoretical findings are further validated through a real-world case study based on ERCOT's 4CP program, demonstrating the practical relevance of the proposed game model. 


Our insights on peak shaving effectiveness have the potential to inform utilities and operators that are concerned with applying game models. By analyzing the efficiency loss in quasiconcave and non-concave games within the two-agent setting, we show that efficiency decreases with the disparities between agents, as measured by their marginal shifting costs. This insight informs the design of CP shaving mechanisms that simultaneously achieve effectiveness and fairness by aligning marginal shifting costs - either by incentivizing non-flexible agents with higher comfort loss or regulating them to shift less demand. We show that efficiency decreases as the game transitions from concave to non-concave and then to quasiconcave, corresponding to greater agent flexibility. This suggests that higher flexibility amplifies system inefficiency, which contradicts the common perception that utilities and operators generally favor flexible agents for their ability to provide grid services. Combined with our multi-agent analysis, which shows that efficiency loss is sensitive to agent flexibility in small systems but stabilizes in larger ones, we suggest that forming large systems with flexible agents while grouping inflexible agents into smaller systems can effectively increase efficiency.


Our work has a few limitations. First of all, we formulate the demand-shifting penalty as a quadratic function. Future research is needed to examine more generic functions and even obtain empirical models from data-driven methods. Another limitation is that we only look at a two-period scheduling problem with simultaneous decisions. Additional research is needed to study the CP shaving game under a multi-period or a sequential context, where in the latter, agents demand shifting decisions are made stage by stage given the non-anticipatory price and demand realization. This also motivates the last limitation, which is that we did not consider the incomplete information of the game model. In practice, agents' payoff structure should be private information that can't be observed by others, but agents may get others' previous decisions and system peak period to update their beliefs gradually. We would expect contextual optimization to be a promising way to infer the private payoff function from other observed features and embed it into the stochastic formulation.
\bibliographystyle{apalike}
\bibliography{Ref}

\newpage

\begin{center}
    \setlength{\baselineskip}{2.0\baselineskip} 
    {\fontsize{23}{30}\selectfont  Appendix: Gaming on Coincident Peak Shaving: Equilibrium and Strategic Behavior} 
\end{center}

\appendices
\renewcommand{\thesection}{Appendix \Alph{section}}

\section{Proof of Proposition \ref{game_proper}}
\begin{proof}
We separate two conditions to analyze the continuity and concavity due to the indicator function. 

(1) We begin by analyzing the concave game conditions. For the game to be concave, the CP period must always align with the baseline CP period throughout the game. This allows us to eliminate the indicator function, reducing the payoff function to the quadratic form described in (\ref{1a}), which is concave and continuous. Since altering the CP period requires the total demand shifting of all agents to exceed the system balance point $b$, and the maximum demand shifting within one period is limited by their respective critical points, the concavity condition is given by $b > r_{\mathrm{x}} + r_{\mathrm{y}}$. Similarly, if baseline CP period is 1, the condition is given by $b <- r_{\mathrm{x}} -r_{\mathrm{y}}$. 


(2) We now establish the conditions for quasiconcave/discontinuous games. When the CP period shifts during the gaming process, the presence of the indicator function clearly introduces discontinuity in the game. To analyze this, we express agent x's payoff function in terms of the switching point $c_\mathrm{x} = b - y$
\begin{subequations}
    \begin{align}
        f_{\mathrm{x}}(x,y) &= -(X_{1}+x)I(x,y)-(X_{2}-x)(1-I(x,y)) - \alpha_{\mathrm{x}}x^2, \label{switching_point}\\
        &I(x,y) = \begin{cases} 
      1 & x-c_\mathrm{x}\geq 0 \\
      0 & x-c_\mathrm{x} < 0 
        \end{cases},\label{step_function}
    \end{align}             
\end{subequations}
     
According to the definition of quasiconcavity, for all $x',x'' \in \mathcal{X}_{\mathrm{x}}$ and $\lambda \in[0,1]$, agent x's payoff function $f_{\mathrm{x}}(x,y)$ should satisfy the following for all $y\in \mathcal{X}_{\mathrm{y}}$,
    \begin{align}
         f_\mathrm{x}(\lambda x' + (1-\lambda)x'',y)\geq \min \{f_\mathrm{x}(x',y), f_\mathrm{x}(x'',y)\}.\label{quasi_concave_property}
    \end{align} 
    This quasiconcavity property is influenced by the switching point $c_\mathrm{x}$ and the critical point $r_{\mathrm{x}}$, leading to three possible cases for further analysis. Note that although by our consumption, the demand shifting $x$ of agent x is non-negative, the demand shifting $y$ of agent y is unrestricted. 
    
    i) Agent x can switch the CP period regardless of whether it occurs in period 1 or 2, meaning its switching point does not exceed its critical point in either period, i.e., $-r_{\mathrm{x}} \leq c_\mathrm{x} \leq r_{\mathrm{x}}$. This leads to the following conditions:
      \begin{align}
         -r_{\mathrm{x}} \leq b-y,\ r_{\mathrm{x}} \geq b-y,\label{9b}   
      \end{align}  
    Graphically, the payoff function is monotonically increasing before the switching point and decreasing afterward, forming a peak at $c_\mathrm{x}$
    
    ii) Agent x can switch the CP period only if it occurs in period 1, meaning its switching point exceeds its critical point in period 2 but not in period 1, i.e., $r_{\mathrm{x}} < c_\mathrm{x}$. Additionally, the function satisfies the l.s.c. condition, as defined in Definition \ref{Nash_Equ_defi}, i.e., $-(X_{2}-c_\mathrm{x}) \leq  -(X_{1}+c_\mathrm{x})$. This implies that
    \begin{align}
        &r_{\mathrm{x}} < b - y,
        \ \frac{X_{2}-X_{1}}{2} = b_{\mathrm{x}} \leq  b - y, \label{9d}
    \end{align}
    Graphically, the function exhibits quadratic concavity to the left of the switching point, monotonic decrease to the right, and l.s.c. at the switching point.

    iii) Agent x can switch CP period only if it occurs in period 2, meaning its switching point exceeds its critical point in period 1 but not in period 2, i.e., $-r_{\mathrm{x}} > c_\mathrm{x}$. Also, the function satisfies the u.s.c. conditions as defined in Definition \ref{Nash_Equ_defi}, i.e., $-(X_{2}-c_\mathrm{x}) \leq -(X_{1}+c_\mathrm{x})$. This implies that
    \begin{align}
        &-r_{\mathrm{x}} > b - y,\ 
        \frac{X_{2}-X_{1}}{2} = b_{\mathrm{x}} \geq b - y,   \label{9f}
    \end{align}
    Graphically, the function exhibits quadratic concavity to the right of the switching point, monotonic decrease to the left, and u.s.c. at the switching point.

    By combining the three cases, we require that at least one of them holds for both agents' payoff functions across all possible strategies of the other agent to ensure the game remains quasiconcave. From the graphical representation, we observe that if an agent's payoff function satisfies case i), the optimal strategy is at its switching point $x=c_\mathrm{x},y=c_\mathrm{y}$. If an agent's payoff function satisfies case ii) or iii), the optimal strategy is at its critical point, specifically $x=r_{\mathrm{x}}, y=r_{\mathrm{y}}$ for case ii) and $x=-r_{\mathrm{x}}, y=-r_{\mathrm{y}}$ for case iii). Now, suppose agent x's payoff function satisfies case i). According to condition (\ref{9b}), agent y's strategy needs to satisfy both $y \geq b-r_{\mathrm{x}}$ and $y \leq b+r_\mathrm{x}$. If agent y's payoff function satisfies case ii) or iii), we obtain either $r_{\mathrm{x}} + r_{\mathrm{y}} \geq b$ or $-r_\mathrm{x} - r_\mathrm{y} \leq b$, which correspond to the concave game conditions as previously described, except when equality holds, i.e., $r_{\mathrm{x}} + r_{\mathrm{y}} = b$ or $-r_\mathrm{x} - r_\mathrm{y} = b$. These equality conditions imply that their critical points equal to their balance points, meaning $b_{\mathrm{x}} = r_{\mathrm{x}}$ and $b_{\mathrm{y}} = r_{\mathrm{y}}$. Otherwise, for the game to be quasiconcave, both agents' payoff functions must satisfy case (i). This implies that both agents can fully balance their demand over the two periods, and their maximum demand shifting are given by $x = b_{\mathrm{x}}, y = b_{\mathrm{y}}$. Substituting these into condition (\ref{9b}), we derive the conditions that ensure quasiconcavity.
    \begin{align}
        -r_{\mathrm{x}}\leq b_{\mathrm{x}} \leq r_{\mathrm{x}}, \ -r_{\mathrm{y}}\leq b_{\mathrm{y}} \leq r_{\mathrm{y}}. \label{39}
    \end{align}
    
    Otherwise, suppose agent x's payoff function satisfies case ii), according to the condition (\ref{9d}), agent y's best strategy needs to satisfy $y < b-r_{\mathrm{x}},y\leq b-b_{\mathrm{x}}$. We know agent y's best strategy is $y =-r_{\mathrm{y}}$ or $y=r_{\mathrm{y}}$ when its payoff function satisfies cases ii) or iii), respectively. As $r_{\mathrm{y}}$ are independent to $r_{\mathrm{x}}$ and $b_{\mathrm{x}}$, agent y's payoff function satisfies cases ii) or iii) can't guarantee the quasiconcave conditions. Thus, we conclude that each agent's payoff function is quasiconcave if (\ref{39}) is true, which also shows each agent is capable according to Definition \ref{capably}.
    Given our baseline assumption, agent x's condition can further simplified to $0 \leq b_{\mathrm{x}} \leq r_{\mathrm{x}}$.
    
    
    (3) Other than the above two conditions, the game is non-concave and still discontinuous due to the indicator function. In this case, at least one agent is non-capable, and they can change the CP period together. We can directly write the conditions as the complementary set of concave and quasiconcave conditions with respect to $\mathbb{R}$, i.e., (\ref{non-concave}). Also, note that if one agent's payoff function is quasiconcave or non-concave, the payoff functions of the other agents must also be quasiconcave or non-concave, which excludes the conditions that one agent's payoff function is concave and the other is quasiconcave. This finishes the proof of this proposition.
\end{proof}

\section{Proof of Theorem \ref{n_agent_game}}
\emph{Overview of the proof}: We begin by deriving a Lemma to establish the NE for the quasiconcave CP shaving game $G$, providing insight into the best response rationale. We then proceed to prove the theorem by analyzing the shifting capabilities of both agents.
\begin{lemma}\emph{NE under quasiconcave two-agent setting.}\label{lemma_NE2agent}
    Under the quasiconcave two-agent, two-period setting as described in Proposition \ref{game_proper}, the NE of the CP shaving game, as defined in Definition \ref{Nash_Equ_defi}, is given by $x^* = b_\mathrm{x}, y^* = b_\mathrm{x}$ and satisfies the condition $S_1(x^*,y^*) = S_2(x^*,y^*)$.
\end{lemma}
\begin{proof}
    Since we are dealing with the quasiconcave condition, we assume $x$ is non-negative, while $y$ remains unrestricted. Thus, we conduct a comprehensive analysis without imposing specific baseline demand assumptions.
    From the best response perspective, for a given instance, suppose the CP period is 1, meaning $S_1 = S_2+ \delta$, where $\delta$ is an infinitesimal number. If agent x benefits by reducing $x$, shifting the CP period to 2, this implies that agent x's demand in period 1 is greater than in period 2, i.e., $X_{1}+x \geq X_{2}-x$. Now, considering agent y, if its demand in period 1 is lower than in period 2 ($Y_{1}+y < Y_{2}-y$), the CP period shift from 1 to 2 harms agent y, leading them to increase $y$ in an attempt to push the CP period back to 1. Conversely, if $Y_{1}+y \geq Y_{2}-y$, shifting the CP period to 2 benefits agent y as it can shift $y$ away from period 2, meaning its best response remains to increase $y$. Thus, regardless of agent y's demand distribution, its best response is always adversarial to agent x.
    
    This demonstrates that both agents are  \emph{fully competitive}—if agent x benefits from changing their strategy, that same strategy directly harms agent y. This mirrors a zero-sum game setting, where one agent's gain is the other's loss. Given this competitive structure, we can apply the following min-max formulation to analyze their best responses.
        \begin{subequations}
            \begin{align}
               x^* &= \arg\max_{x} \min_{y} f_\mathrm{x}(x,y) = \arg\max_{x} 
               \left\{\begin{aligned}
           -\pi(X_{1}&+x)-\alpha_{\mathrm{x}}x^2,   \\
           &x \geq b_{\mathrm{x}}, x \geq b-y^*\\
           -\pi(X_{2}&-x)-\alpha_{\mathrm{x}}x^2,   \\
           &x < b_{\mathrm{x}}, x< b-y^*
        \end{aligned}.  \right.
        \end{align}
        where $y^*$ is the best strategy of agent y following the same structure, and the first and second case corresponds to CP periods 1 and 2. By applying the first-order optimality condition, the best strategy of $x$ is 
        \begin{align}
            x^* = \left\{\begin{aligned} 
            &\max\{-r_{\mathrm{x}}, b-y^*, b_{\mathrm{x}}\}, 
            &S_{1}\geq S_{2}\\
            &\min\{r_{\mathrm{x}}, b-y^*, b_{\mathrm{x}}\}, 
            &S_{1} < S_{2}\\
            \end{aligned}.  \right. \label{x_1^*}
        \end{align}
        The best strategy for agent y follows the same structure
        \begin{align}
            &y^* = \arg\max_{y} \min_{x} f_{\mathrm{y}}(x,y) 
            = \left\{\begin{aligned} 
            &\max\{-r_{\mathrm{y}}, b - x^*, b_{\mathrm{y}}\}, 
            &S_{1}\geq S_{2} \\
            &\min\{r_{\mathrm{y}}, b-x^*, b_{\mathrm{y}}\}, 
            &S_{1} <S_{2} \\
            \end{aligned}.  \right.\label{x_2^*}
        \end{align}
        \end{subequations}

        According to (\ref{quasi_discounti}), we know $-r_{\mathrm{x}} \leq b_{\mathrm{x}}\leq r_{\mathrm{x}}$, $-r_{\mathrm{y}} \leq b_{\mathrm{y}}\leq r_{\mathrm{y}}$, then the (\ref{x_1^*}) and (\ref{x_2^*}) can be simplified to $y^* = \max\{ b-x^*, b_{\mathrm{y}}\}$ and $x^* = \max \{b-y^*, b_{\mathrm{x}}\}$, which shows the only mutual best strategy are $x^* = b_{\mathrm{x}},\ y^* = b_{\mathrm{y}}$ (the NE), and the conditions are 
        \begin{align}
            X_{2} + Y_{2} - x^* - y^* = S_2
            = X_{1} + Y_{1} +x^* + y^* = S_1.
        \end{align}

        This shows the NE will always be obtained when $S_1 = S_2$ under the quasiconcave condition. Essentially, during the game, periods 1 and 2 are active interactively and finally converges at the connecting point of both periods, i.e., $S_1=S_2$. This means there is a unique equilibrium point in the entire strategy set $\mathcal{X}$. Indeed, all agents will minimize the payment associated with $S_1, S_2$ imbalance as any imbalance results in a significant CP charge change caused by the opponent's strategy.
\end{proof}

This lemma establishes the NE under the quasiconcave game condition and highlights that agents' best response is to adopt adversarial strategies when their individual peak periods differ to reduce their CP charge. In contrast, the concave game is straightforward to analyze, as both agents maximize quadratic independent payoff functions, making their best strategies simply their critical points of the CP period. Building on this foundation, we now present the proof of the theorem.

\begin{proof}[Proof of Theorem \ref{n_agent_game}]
The basic idea behind proving this theorem is to analyze whether agents are capable of fully balancing their demand. If they are not, we further determine whether they are upper non-capable or lower non-capable, as each condition corresponds to a different NE. The key rationale is that both agents' best response is to shift demand to reduce CP charge, but their shifting capability imposes a constraint. Thus, the key aspect of the proof lies in analyzing the relationship between their critical point and balance point.

Given our assumption that $S_{\mathrm{b},1}<S_{\mathrm{b},2}$ and $X_1<X_2$, if agent x is non-capable, it must be upper non-capable, satisfying $b_{\mathrm{x}}>r_{\mathrm{x}}$. This corresponds to scenario ii) as described in Proposition \ref{game_proper}. Graphically, this means that agent x's payoff function is quadratic (concave) in the left part of the switching point, monotonically decreasing in the right part, but does not satisfy l.s.c. at the switching point. Meanwhile, agent y remains unrestricted in its decision-making.

We then separate many scenarios according to whether they are non-capable agents, upper non-capable, or lower non-capable, to analyze the NE and corresponding conditions. 


(1) $X_{1} < X_{2}, Y_{1} < Y_{2}$, meaning both agents have higher individual demand in the baseline CP period, which is period 2 by assumption. The analysis is divided into four scenarios, depending on whether each agent is capable or non-capable.
\begin{itemize}
    \item Both agents are capable agents, i.e., $0 \leq b_{\mathrm{x}} \leq r_{\mathrm{x}}, 0 \leq b_{\mathrm{y}} \leq r_{\mathrm{y}}$, the game is quasiconcave, meaning both agents can fully balance their demand over two periods. By Lemma \ref{lemma_NE2agent}, the NE is $x^* = b_{\mathrm{x}}, y^* = b_{\mathrm{y}}$.
    \item Agent x is non-capable while agent y is capable, i.e., $b_{\mathrm{x}} > r_{\mathrm{x}} >0, 0\leq b_{\mathrm{y}} \leq r_{\mathrm{y}}$, agent x must be upper non-capable since $X_{1} < X_{2}$, and its best strategy is to shift demand up to its critical point in period 2, so, $x^* = r_{\mathrm{x}}$. Since agent y is capable, it can further shift demand to help balance the system. Essentially, agent y can reduce costs by shifting demand away from CP period 2 but must do so within its shifting capability. The optimal strategy for agent y depends on whether it can fully balance the system demand given agent x's strategy. Thus, the best strategy for agent y is $y^* = \min \{r_{\mathrm{y}}, b-x^*\}$. This distinction determines whether the game remains concave or becomes non-concave, leading to two scenarios. 
    To determine the scenario, we compare the two terms. i) If $ b-x^* = b - r_{\mathrm{x}} > r_{\mathrm{y}}$, meaning agent y cannot shift enough to change the CP period given agent x's strategy, then the game remains concave, satisfying $r_{\mathrm{x}} +r_{\mathrm{y}} < b$. In this case, the best strategy is $y^* = r_{\mathrm{y}}$, with the conditions $r_{\mathrm{x}} < b_{\mathrm{x}},b_{\mathrm{y}} \leq r_{\mathrm{y}}$ and $r_{\mathrm{x}} +r_{\mathrm{y}} - b_{\mathrm{x}} < b_{\mathrm{y}}$
    ii) Otherwise, when $r_{\mathrm{x}} +r_{\mathrm{y}} \geq b$, agent y can shift enough demand to change the CP period, making the game non-concave. In this scenario, agent y's best strategy is $y^* = b-x^* = b-r_\mathrm{x}$. The corresponding condition is $r_{\mathrm{x}} < b_{\mathrm{x}}, b_{\mathrm{y}} \leq r_{\mathrm{y}}$ and $b_{\mathrm{y}} \leq r_{\mathrm{x}} +r_{\mathrm{y}} - b_{\mathrm{x}}$.
    \item Agent x is capable while agent y is non-capable, i.e., $0\leq b_{\mathrm{x}} \leq r_{\mathrm{x}}, b_{\mathrm{y}} > r_{\mathrm{y}}$. Similar to prior scenario, agent y can only be upper non-capable. Thus, by symmetry, we can interchange x and y to derive the NE and corresponding conditions. The equilibrium strategies are $y^* = r_\mathrm{y}, x^* = r_\mathrm{x}$ under the condition of $b_{\mathrm{x}} \leq r_{\mathrm{x}},b_\mathrm{y}>r_\mathrm{y}$ and $r_{\mathrm{x}} +r_{\mathrm{y}} - b_{\mathrm{y}} < b_{\mathrm{x}} \leq r_{\mathrm{x}}$, while $y^* = r_\mathrm{y}, x^* = b-r_\mathrm{y}$ under the condition of $b_{\mathrm{x}} \leq r_{\mathrm{x}},b_\mathrm{y}>r_\mathrm{y}$ and $b_{\mathrm{x}} \leq r_{\mathrm{x}} +r_{\mathrm{y}} - b_{\mathrm{y}}$;
    \item Both agents are (upper) non-capable, i.e., $b_{\mathrm{x}} > r_{\mathrm{x}}, b_{\mathrm{y}} > r_{\mathrm{y}}$, the game is concave and the NE of both agents are their critical point $x^* = r_{\mathrm{x}},y^* = r_{\mathrm{y}}$.
       
\end{itemize}

(2) $X_{1} < X_{2}, Y_{1} \geq Y_{2}$, agent x has higher individual demand in the baseline CP period, while agent y has lower demand. We again consider four scenarios.
\begin{itemize}
    \item Both agents are capable, i.e., $0 \leq b_{\mathrm{x}} \leq r_{\mathrm{x}}, -r_{\mathrm{y}} \leq b_{\mathrm{y}} \leq 0$, similar to the first scenario in case (1), the NE remains the same as in that scenario.
    \item Agent x is upper non-capable while agent y is capable, i.e., $b_{\mathrm{x}} > r_{\mathrm{x}}, -r_{\mathrm{y}}\leq b_{\mathrm{y}} \leq 0$. Similar to the second scenario in case (1), agent x's best strategy remains $x^* = r_\mathrm{x}$, but agent y's best strategy differs due to the baseline demand conditions, given by $y^* = \max\{-r_\mathrm{y}, b-x^*\}$. From the conditions of this scenario, we have $-r_\mathrm{y} - b_{\mathrm{y}}\leq 0, r_\mathrm{x} - b_\mathrm{x} < 0$, leading to $-r_\mathrm{y} + r_\mathrm{x} - b_{\mathrm{y}} - b_\mathrm{x} = -r_\mathrm{y} + r_\mathrm{x} - b<0$. This implies $-r_\mathrm{y} < b-r_{\mathrm{x}}$, so the best strategy for agent y is $y^* = b-r_{\mathrm{x}}$. The corresponding conditions are $b_{\mathrm{x}} > r_{\mathrm{x}},-r_{\mathrm{y}}\leq b_{\mathrm{y}}$ and $-r_\mathrm{y} + r_\mathrm{x} - b_\mathrm{x}<b_\mathrm{y}$.
    \item Agent x is capable while agent y is non-capable, i.e., $0\leq b_\mathrm{x} \leq r_\mathrm{x}, -r_\mathrm{y} > b_\mathrm{y}$, agent y must be lower non-capable since $Y_{1} \geq Y_{2}$, and its best strategy is to shift demand up to its critical point in period 1, so $y^* = -r_\mathrm{y}$. Similar to the analysis in the second scenario in case (1), agent x's best strategy is determined by whether it can fully balance the system demand given agent y’s strategy, i.e., $x^*=\max\{-r_\mathrm{x},b-y^*\}$. Since $X_1<X_2$, we must have $x \geq 0$, thus $x^* = b-y^* = b+r_\mathrm{y}$, and the conditions for agent x is $0\leq b_\mathrm{x} \leq r_\mathrm{x}$, for agent y is $b+ r_\mathrm{y} \geq -r_\mathrm{x}$, equivalent to $b_\mathrm{y} \geq -r_\mathrm{x} -r_\mathrm{y} -b_\mathrm{x}$
    \item Both agents are non-capable, i.e., $b_{\mathrm{x}} > r_{\mathrm{x}}, -r_{\mathrm{y}} > b_{\mathrm{y}}$. Given their baseline conditions, both agents must compete, and their best strategy is determined by which agent can shift more demand before reaching their critical point. The equilibrium strategy is thus given by
    \begin{align}
        x^* = \min\{r_{\mathrm{x}},b-y^*\}, y^* = \max\{-r_{\mathrm{y}},b-x^*\}.
    \end{align}
    i) If $x^* = r_{\mathrm{x}}$ then it must hold that $r_{\mathrm{x}} < b-y^*$. Suppose $y^* = -r_{\mathrm{y}}$, then for agent y, we require $-r_{\mathrm{y}} > b-x^* = b-r_\mathrm{x}$. This results in the conditions $r_{\mathrm{x}} < b-y^* = b+ r_{\mathrm{y}}$ for agent x. Obviously, the condition of agents x and y are conflict. Thus, the NE for both agents is given by $x^* = r_{\mathrm{x}}, y^* = b -x^* =b-r_{\mathrm{x}}$ with the corresponding condition $r_{\mathrm{x}} < b_\mathrm{x},-r_{\mathrm{y}} > b_\mathrm{y}$ and $-r_{\mathrm{y}} < b-r_\mathrm{x}$.
    ii) Otherwise, if $r_{\mathrm{x}} > b-y^*$, then $x^* = b-y^*$. Suppose $y^* = b-x^*$, then for agent y, we require $-r_{\mathrm{y}} < b-x^*$, which simplifies to $r_{\mathrm{x}} > b-y^* =x^*$ for agent x and $-r_{\mathrm{y}} < b-x^* = y^*$ for agent y. This implies that agent x does not reach the critical point in period 2, and agent y does not reach the critical point in period 1. Therefore, their best strategies would be the balance points, i.e., $x^* =b_\mathrm{x}, y^* = b_\mathrm{y}$. However, this contradicts the conditions of this scenario that $b_\mathrm{x} > r_\mathrm{x}, -r_\mathrm{y} > b_\mathrm{y}$. Thus, the NE for both agents are $x^* =b-y^* = b+r_{\mathrm{y}}, y^* = -r_{\mathrm{y}}$. The correspond conditions are $r_{\mathrm{x}} < b_\mathrm{x},-r_{\mathrm{y}} > b_\mathrm{y}$ and $, -r_{\mathrm{y}} > b-r_{\mathrm{x}}$.
\end{itemize} 

To conclude, combining all the scenarios, we obtain the NE and corresponding solutions for different game structures. 

The game is \emph{quasiconcave} when both agents are capable, i.e., $0 \leq b_\mathrm{x}\leq r_\mathrm{x}$ for agent x and $-r_\mathrm{y} \leq b_\mathrm{y} \leq r_\mathrm{y}$ for agent y, which aligns with (\ref{quasi_discounti}) as described in Proposition \ref{game_proper}. The NE in this case is $x^* =b_\mathrm{x}, y^* = b_\mathrm{y}$.

The game is \emph{concave} if the sum of agents’ shifting capacities is insufficient to change the CP period, which is the union of the following 
\begin{subequations}
    \begin{align}
        r_\mathrm{x} < b_\mathrm{x},\ r_{\mathrm{x}} +r_{\mathrm{y}} - b_{\mathrm{x}} < b_{\mathrm{y}} \leq r_{\mathrm{y}},\\
        b_\mathrm{y} > r_\mathrm{y},\ r_{\mathrm{x}} +r_{\mathrm{y}} - b_{\mathrm{y}} < b_{\mathrm{x}} \leq r_{\mathrm{x}},\\
        r_\mathrm{x} < b_\mathrm{x}, r_\mathrm{y} < b_\mathrm{y}.
    \end{align}
\end{subequations}
This is equivalent to $r_{\mathrm{x}} +r_{\mathrm{y}} < b$ and the same to (\ref{concave}) as described in Proposition \ref{game_proper}. The NE in this case is $x^* = r_{\mathrm{x}}, y^* = r_{\mathrm{y}}$.

The game is \emph{non-concave} when $0\leq b \leq r_{\mathrm{x}} +r_{\mathrm{y}}$. We have three cases for the NE and corresponding conditions. i) $x^* =r_\mathrm{x},y^* = b-r_\mathrm{x}$, the conditions are the union of the following 
\begin{subequations}
    \begin{align}
        r_\mathrm{x} < b_\mathrm{x},
        \ b_\mathrm{y} \leq r_\mathrm{y},
        \ b_{\mathrm{y}} \leq r_{\mathrm{x}}+ r_{\mathrm{y}} - b_{\mathrm{x}},\\
        r_\mathrm{x} < b_\mathrm{x},
        \ -r_{\mathrm{y}}\leq b_{\mathrm{y}}
        \ b_{\mathrm{y}} > r_{\mathrm{x}} - r_{\mathrm{y}} - b_{\mathrm{x}},\label{38b}\\
        r_\mathrm{x} < b_\mathrm{x},
        \ b_{\mathrm{y}} < -r_{\mathrm{y}},
        \ b_{\mathrm{y}} > r_{\mathrm{x}} - r_{\mathrm{y}} - b_{\mathrm{x}},\label{38c}
    \end{align}
    which is equivalent to $0\leq b\leq r_{\mathrm{x}}+ r_{\mathrm{y}}$ and $r_\mathrm{x} < b_\mathrm{x}$. The reason is that given this conditions, automatically, we have $-r_{\mathrm{y}}\leq b_{\mathrm{y}}$ and indicates $b_{\mathrm{y}} > r_{\mathrm{x}} - r_{\mathrm{y}} - b_{\mathrm{x}}$ (\ref{38b}) and (\ref{38c}).
\end{subequations}
ii) $x^* =b-r_\mathrm{y},y^* = r_\mathrm{y} $, the conditions are $b<r_{\mathrm{x}}+ r_{\mathrm{y}}$ and $r_\mathrm{y} < b_\mathrm{y}$.
iii) $x^* =b+r_\mathrm{y},y^* =- r_\mathrm{y} $, the conditions are the union of the following 
\begin{subequations}
    \begin{align}
        b_\mathrm{x} \leq r_\mathrm{x},
        \ -r_\mathrm{y} > b_\mathrm{y},
        \ b_{\mathrm{y}} \geq -r_{\mathrm{x}}- r_{\mathrm{y}} - b_{\mathrm{x}},\\
        r_\mathrm{x} <b_\mathrm{x},
        \ -r_\mathrm{y} > b_\mathrm{y},\ b_{\mathrm{y}} \leq r_{\mathrm{x}}- r_{\mathrm{y}} - b_{\mathrm{x}},\label{39b}
    \end{align}
  which is equivalent to $0\leq b \leq r_{\mathrm{x}} + r_{\mathrm{y}}$ and $-r_\mathrm{y} > b_\mathrm{y}$. Still, the reason is that given this conditions, automatically, we have $b_\mathrm{x} <r_\mathrm{x} + 2r_\mathrm{y}$, indicating $b_{\mathrm{y}} \leq r_{\mathrm{x}}- r_{\mathrm{y}} - b_{\mathrm{x}}$ (\ref{39b})
\end{subequations}
These results are consistent with (\ref{non-concave}) in Proposition \ref{game_proper} and (\ref{mix2}) in this theorem. This completes the proof.
\end{proof}

\section{Proof of Theorem \ref{dynamics_theorem}}
\emph{Overview of the proof}:
    The basic idea of proving this Theorem is to show that the dynamical systems of each period (\ref{dynamics_system1}) and (\ref{dynamics_system2}) are asymptotically stable in the strategy set $\mathcal{X}$, then, add the switching logic to show the system is global uniform asymptotically stable in $\mathcal{X}_{\mathrm{s}}$ as described in Theorem \ref{dynamics_theorem}, i.e., local uniform asymptotically stable in $\mathcal{X}$. We specify the proof process as follows:
    \begin{itemize}
        \item We first prove dynamical systems of each period (\ref{dynamics_system1}) and (\ref{dynamics_system2}) is asymptotically stable in the strategy set $\mathcal{X}$. (Lemma \ref{lemma_local_stable}).
        \item Then we prove the overall system (\ref{switched_system}) with the switching logic is global uniform asymptotically stable in $\mathcal{X}_{\mathrm{s}}$. 
    \end{itemize}
 
We first provide Lemma \ref{lemma_local_stable} to show the system stability within each period.
\begin{lemma}\emph{System stability of one period.} \label{lemma_local_stable}
    The system dynamics in period 1 described by (\ref{dynamics_system1}) is asymptotically stable in $\mathcal{X}$, i.e., for every starting point $x,y\in \mathcal{X}$, the solution $(x(k),y(k))$ to the (\ref{dynamics_system1}) converges to an equilibrium point $(x^*,y^*)$ as $k \rightarrow \infty$, where $F_1(x^*,y^*)=0$.
\end{lemma}
\begin{proof}
    The key is to show the rate of change of $\lVert F_1(x,y)\rVert^2$ is always negative for $F_1(x,y)\neq 0$~\citep{rosen1965}. We have 
        \begin{align}
            \frac{dF_1}{dk} &= G \frac{dx}{dk} = G\dot{x},\label{sys_dynamic_deriv}
        \end{align}
        where $G$ is the Jacobian of $F_{1}(x,y)$, and $G = -2\text{diag}(\alpha_\mathrm{x},\alpha_\mathrm{y})$, where $\text{diag}(\cdot):\mathbb{R}^2 \to \mathbb{R}^{2 \cdot 2}$.

        \begin{subequations}
        Now, according to (\ref{dynamics_system1}) and combining with the (\ref{sys_dynamic_deriv}), we have
        \begin{align}
        \frac{1}{2}\frac{d\lVert F_1^2\rVert}{dk} &= \frac{1}{2}\frac{dF_1^T F_1}{dk} = F_1^T \frac{dF_1}{dk} = F_1^T G F_1 = \frac{1}{2}F_1^T(G + G^T)F_1. \label{f_norm_detivative}
        \end{align}
    Because the $G+G^T$ is negative definite, we conclude that, for some $\epsilon>0$, (\ref{f_norm_detivative}) is equivalent to 
    \begin{align}
        \frac{1}{2}\frac{d\lVert F_1^2\rVert}{dk} \leq -\epsilon \lVert F_1\rVert^2 \label{final_result_lemma_local}
    \end{align}        
    Thus, $\lim_{k\rightarrow \infty}\lVert F_1\rVert = 0$, so that $(x(k),y(k))\rightarrow (x^*,y^*)$, where $(x^*,y^*)$ is the equilibrium point and $F_1(x^*,y^*)=0$. Following the interior trajectory theorem from~\citep{rosen1965}, we know $(x^*,y^*) \in \mathcal{X}$, which proves this Lemma by showing system (\ref{dynamics_system1}) is asymptotically stable.
     \end{subequations}
\end{proof}

From Lemma \ref{lemma_local_stable}, the system (\ref{dynamics_system1}) is asymptotically stable in $\mathcal{X}$, and this asymptotically stable result can be extrapolated to the dynamical systems of period 2 described by (\ref{dynamics_system2}).
We then add the switched logic to study global uniform asymptotically stability at the equilibrium point satisfied (\ref{equilbrium_point}) in $\mathcal{X}_{\mathrm{s}}$ and prove the Theorem.

\begin{proof}[Proof of Theorem \ref{dynamics_theorem}]
From Lemma \ref{lemma_local_stable}, all system dynamics within each period are asymptotically stable in the strategy set $\mathcal{X}$. We then rewrite the gradient of the payoff functions in the two periods as follows, 
\begin{subequations}
\begin{align}
    &F_1 = A[x,y]^T+C_1, F_2 = A[x,y]^T + C_2,\\
    &A = [-2\alpha_{\mathrm{x}}, 0; 0, -2\alpha_{\mathrm{y}}], C_1 = -[\pi, \pi]^T, C_2 = -C_1.
\end{align}   
\end{subequations}

We prove the Theorem based on the multiple Lyapunov function method~\citep{switching_dynamics}. Since our system dynamics is liner, the basic idea is to (i) find the functions $\mathcal{V}_j>0$ in each system (period) $j=1,2$, for $\dot{\mathcal{V}}_j \neq 0$, the function $\mathcal{V}_j$ is always decreased along the solution of the $j$th system in the region where this system is active, i.e., it is peak;
(ii) on the switching surface $S_1(x,y)=S_2(x,y)$ the function $\mathcal{V}_j$'s value match.  

    We then choose $\mathcal{V}_j, j= 1,2$ as follows: 
    \begin{subequations}
    \begin{align}
    &\mathcal{V}_1 = -[x,y] \frac{A}{2} [x,y]^T - C_1^T [x,y]^T +d_1,
    \mathcal{V}_2 = -[x,y] \frac{A}{2} [x,y]^T - C_2^T [x,y]^T +d_2, \\
    &d_1 = \pi S_{\mathrm{b},1}, 
    d_2 = \pi S_{\mathrm{b},2}.
        \end{align}
    \end{subequations}
    Note that $-\partial \mathcal{V}_1/\partial [x,y]^T = F_1(x,y)$ and $-\partial \mathcal{V}_2/\partial [x,y]^T = F_2(x,y)$.
    
       We first show the regions that guarantees function $\mathcal{V}_j>0$ 
       \begin{subequations}\label{local_range}
           \begin{align}
               \alpha_{\mathrm{x}} x^2 + \alpha_{\mathrm{y}} y^2 + \pi(x+y + S_{\mathrm{b},1}) > 0, \\
                \alpha_{\mathrm{x}} x^2 + \alpha_{\mathrm{y}} y^2 + \pi(-x-y + S_{\mathrm{b},2}) > 0,
           \end{align}
       \end{subequations}
       which is the same as (\ref{local_stable}) shows in this Theorem. 
      We then show the rate of change of $\dot{\mathcal{V}}_j$ is negative.
    \begin{subequations}\label{vdot_condition}
      \begin{align}
         &\dot{\mathcal{V}}_1
         = \frac{\partial \mathcal{V}_1}{\partial [x,y]^T} F_1 = -(\frac{A^T+A}{2}[x,y]^T + C_1)(A[x,y]^T+C_1) = [x,y] A' [x,y]^T + [x,y]B' - C_1^TC_1, \\
         &A' =[-4\alpha_{\mathrm{x}}^2,0;0,-4\alpha_{\mathrm{y}}^2], 
         B' = -\frac{A^T + A}{2}C_1 - A^TC_1 =  [-4\pi\alpha_{\mathrm{x}}, -4\pi\alpha_{\mathrm{y}}]^T.
      \end{align}  
    We have the critical point (equilibrium point) when $\partial(\dot{\mathcal{V}}_1)/\partial [x,y]^T=0$
      \begin{align}
        [x',y']^T &= -(A'^T+A')^{-1}  B' = -[\frac{\pi}{2\alpha_{\mathrm{x}}},\frac{\pi}{2\alpha_{\mathrm{y}}}]^T.\label{equilbrium_point_solution}
      \end{align}
      It is easy to see $A'$ is negative definite. To show $\dot{\mathcal{V}}_1<0$ except the equilibrium point, we need to show $\dot{\mathcal{V}}_1(x',y') = 0$, so that other point must less than zero.
      \begin{align}
          \dot{\mathcal{V}}_1(x',y') &= B'^T ((A'^T+A')^{-1})^T A' (A'^T+A')^{-1}  B' - B'^T ((A'^T+A')^{-1})^T B' - C_1^TC_1 \nonumber\\
          &= -\frac{1}{4}B'^T (A'^{-1})^T B' - C_1^TC_1 =0.
      \end{align}
          \end{subequations}

      In terms of $\dot{\mathcal{V}}_2$, we have
          \begin{align}
             \dot{\mathcal{V}}_2
         &= \frac{\partial \mathcal{V}_2}{\partial [x,y]^T} F_2 = -(\frac{A^T+A}{2}[x,y]^T + C_2)(A[x,y]^T+C_2). 
          \end{align}
          Due to $C_2 = -C_1$, the structure of $\dot{\mathcal{V}}_2$ is the same to $\dot{\mathcal{V}}_1$ and we can obtain $\dot{\mathcal{V}}_2 <0$ except the equilibrium (critical) point $[x',y']^T = [\pi/2\alpha_{\mathrm{x}}, \pi/2\alpha_{\mathrm{y}}]^T$, where $\dot{\mathcal{V}}_2 (x',y')= 0$.

        Then, we show $\mathcal{V}_1 = \mathcal{V}_2$ on the switching surface $S_1(x,y) =S_2(x,y)$, where $x + y= b$ and $f_{\mathrm{x},1} + f_{\mathrm{y},1}= f_{\mathrm{x},2} + f_{\mathrm{y},2}$, and the key is to show $-C_1^T[x,y]^T+d_1 = -C_2^T[x,y]^T +d_2$,
          \begin{subequations}
          \begin{align}
             -C_1^T[x,y]^T+d_1 &= \pi (x + y) + \pi S_{\mathrm{b},1} =  \pi b + \pi S_{\mathrm{b},1} = \frac{S_{\mathrm{b},2}+S_{\mathrm{b},1}}{2},\\
             -C_2^T[x,y]^T+d_2 &= -\pi (x + y) + \pi S_{\mathrm{b},2} =  -\pi b + \pi S_{\mathrm{b},2} = \frac{S_{\mathrm{b},2}+S_{\mathrm{b},1}}{2}.
          \end{align}
          
        Then, we can conclude that $\mathcal{V}_j$ is always decrease except $\dot{\mathcal{V}}_j=0$. Also, the rate of decrease of $\mathcal{V}_j$ along solutions is not affected by switching, and asymptotic stability is uniform with respect to $j$. 
        
        Now, let's analyze the convergent equilibrium point. If the game is concave, only the baseline CP period 2 is active. Then the function value $\mathcal{V}_2$ decreases over the period until $\dot{\mathcal{V}}_2 =0$, and the system reaches the stable points. According to the Theorem \ref{n_agent_game}, the stable point is the unique equilibrium point as described in (\ref{3NE}). 
        
        Otherwise, both periods will be active sequentially, and $\mathcal{V}_j$ will decrease over time until $\mathcal{V}_1 = \mathcal{V}_2$. The reason is that the dynamics from one period always push the other period active, e.g., $F_1$ in period 1 always decreases $x,y$, which pushes the solution past the switching surface and active period 2. Thus, as $\mathcal{V}_1 = \mathcal{V}_2$ on the switching surface, both periods are finally stable on the switching surface. This means $(x(k),y(k))\rightarrow (x^*,y^*)$ when $k\rightarrow \infty$, where $(x^*,y^*)$ is the stable point. According to the Theorem \ref{n_agent_game}, the stable point should be the unique equilibrium point and satisfy (\ref{1NE}). Thus, we prove the global uniform asymptotically stable of the overall system in (\ref{local_stable}), and the equilibrium point satisfies (\ref{equilbrium_point}). 
        
      \end{subequations}
\end{proof}

    

\section{Proof of Theorem \ref{convergency_theorem}}
\begin{proof}
    The key to proving this Theorem is to select the learning rate based on the backtracking line search method. The learning rate depends on the CP periods on which the current and future steps lie, as well as the payoff functions and the gradients that each agent follows. 

    Suppose the current step is $h$, we express the backtracking line search condition for agent x to choose the learning rate $\tau_{\mathrm{x}}$ for the finite difference approximation (\ref{gradient_algorithm}) as follows (agent y's learning rate $\tau_{\mathrm{y}}$ follow the same rule), 
    \begin{align}
        -f_{\mathrm{x},j} (x_{h+1}) &< -f_{\mathrm{x},j} (x_{h}) - \beta_1 \tau_{\mathrm{x},h} \lVert F_{j}(x_{h},y_h) \rVert^2, \label{line_search}
    \end{align}            
    where $\beta_1$ is the parameter within $[0,0.5]$; $f_{\mathrm{x},j}$ means the function can take either $f_{\mathrm{x},1}$ or $f_{\mathrm{x},2}$ determined by which CP period the current and next step lies, and agent's decision is decoupled within each period. For example, if the current step $x_{h}$ lies in period 1, $j$ will take 1 for step $h$, and if the future step $x_{h+1}$ lies in period 2, $j$ will take 2 for step $h+1$. Note that $f_\mathrm{x}(x)$ from (\ref{agenta}) is formulated as a payoff (profit); here, we use $-f_\mathrm{x}(x)$ to express the cost. 
    
    When selecting the learning rate, we gradually reduce $\tau_{\mathrm{x},h}$ by $\beta_2 \tau_{\mathrm{x},h},\beta_2\in [0,1]$ until (\ref{line_search}) satisfy. If the next step and current step lie on the same CP period, this condition ensures the objective $-f_{\mathrm{x},j} (x_{h})$ reduces by at least $\beta_1 \tau_{\mathrm{x},h} \lVert F_{j}(x_h,y_h)\rVert^2$. This proves the concave game convergence as all the steps lie in the baseline CP period 2, and the objective function is concave, thus, the objective function gradually reduces until $\lVert F_{2}(x_h,y_h)\rVert = 0$.

    In terms of quasiconcave and non-concave games, there are switches during the algorithm iteration. It is easy to imagine that each agent shifts demand away from the baseline CP period at the beginning, monotonically reducing their costs following the cost function from their baseline CP period. Once they reach the balance point (switching surface), $x_{h} + y_{h}=b$, the solution starts switching between two periods. Noted that when switching happens, both agents' individual peak period must be different, i.e., if $X_{1}+x_{h} > X_{2}-x_{h}$ for agent x, then $Y_{1}+y_{h} < Y_{2}-y_{h}$ must hold for agent y; otherwise, there will be no switching;
    
     We then write the difference between $-f_{\mathrm{x},1}(x)$ and $-f_{\mathrm{x},2}(x)$ as 
        \begin{align}
            -f_{\mathrm{x},1}(x) - (-f_{\mathrm{x},2}(x)) =\pi (X_{1} + x -(X_{2} - x)).\label{difference} 
        \end{align}            
     
    Suppose $X_{1}+x_{h} > X_{2}-x_{h}$ for agent x, we know 
    \begin{align}
        -f_{\mathrm{x},1}(x_{h}) > -f_{\mathrm{x},2}(x_{h}), 
        -f_{\mathrm{y},1}(y_{h}) < -f_{\mathrm{y},2}(y_{h}). \label{condition_switch}
    \end{align}
    Now, consider a trajectory starting from period 1, switching to period 2, and back to period 1, i.e., $-f_{\mathrm{x},1}(x_h),-f_{\mathrm{x},2}(x_{h+1}), -f_{\mathrm{x},1}(x_{h+2})$, our goal is to show $-f_{\mathrm{x},1}(x)$ reduce while $-f_{\mathrm{x},2}(x)$ increase for agent x and $-f_{\mathrm{y},2}(y)$ reduce while $-f_{\mathrm{y},1}(y)$ increase for agent y through the trajectory. 
    Due to the switching from period 1 to 2, the gradient in period 1 $F_1(x_h,y_h)$ must be negative to reduce the $(x_h,y_h)$ so that the CP period changes. We choose the learning rate $\tau_{\mathrm{x},h},\tau_{\mathrm{y},h}$ such that
    \begin{subequations}
    \begin{align}
        -f_{\mathrm{x},2}(x_{h+1}) < -f_{\mathrm{x},1}(x_{h}) - \beta_1 \tau_{\mathrm{x},h} \lVert F_1(x_{h},y_h) \rVert^2,\label{agenti}\\
        -f_{\mathrm{y},2}(y_{h+1}) < -f_{\mathrm{y},1}(y_{h}) - \beta_1 \tau_{\mathrm{y},h} \lVert F_1(x_{h},y_h) \rVert^2,\label{agent-i}
    \end{align}       
    \end{subequations}
    As $X_{1}+x_{h} > X_{2}-x_{h}$ and from (\ref{condition_switch}), we know (\ref{agenti}) is easy to be true and we use the corresponding $\tau_{\mathrm{x},h}$ to update $x_{h}$ following (\ref{line_search}); while (\ref{agent-i}) can't be true, so $y_{h}$ will not be updated. Thus, the switching is caused by the update of the agent x's decision, and we know $x_{h+1} < x_{h}$. Because the gradient is in period 1, $F_1<0$, we know the $x_{h+1} > -r_{\mathrm{x}}$ and suppose to be reduced to reach the critical point in period 1 until converge, which indicates the right part of the critical point in the objective function $-f_{\mathrm{x},1}$, where its gradient $-f'_{\mathrm{x},1}>0$. Thus, we have $-f_{\mathrm{x},1}(x_{h+1}) < -f_{\mathrm{x},1}(x_{h})$.
    
    In period 2, according to the trajectory, the gradient will push the solution back to period 1, which requires $x,y$ increase, and thus, we know $F_2(x_{h+1},y_{h+1})>0$. We then choose the learning rate $\tau_{\mathrm{x},h+1},\tau_{\mathrm{y},h+1}$ such that
    \begin{subequations}
    \begin{align}
        -f_{\mathrm{x},1}(x_{h+2}) < -f_{\mathrm{x},2}(x_{h+1}) - \beta_1 \tau_{\mathrm{x},h+1} \lVert F_2(x_{h+1},y_{h+1}) \rVert^2,\label{agenti1}\\
        -f_{\mathrm{y},1}(y_{h+2}) < -f_{\mathrm{y},2}(y_{h+1}) - \beta_1 \tau_{\mathrm{y},h+1} \lVert F_2(x_{h+1},y_{h+1}) \rVert^2,\label{agent-i1}
    \end{align}       
    \end{subequations}
    Still, from (\ref{condition_switch}), we know (\ref{agenti1}) can't be true and (\ref{agent-i1}) is easy to realized by setting $\tau_{\mathrm{y},h+1}$. Thus, the $x_{h+1}$ will not be updated and $y_{h+1}$ will be updated and push the CP period back to 1. Following a similar analysis, we know $y_{h+2} > y_{h+1}$ and $y_{h+2}<r_{\mathrm{y}}$ and are supposed to increase to reach the critical point in period 2 until converge, which indicates the left part of the critical point in the objective function $-f_{\mathrm{y},2}$, where its gradient $-f'_{\mathrm{y},2}<0$. Thus, we have $-f_{\mathrm{y},2}(y_{h+2}) < -f_{\mathrm{y},2}(y_{h+1})$.

    Now, let's look at the entire trajectory, we have $x_{h+2} =x_{h+1}<x_{h}$ for agent x and $y_{h+2} > y_{h+1} = y_{h}$ for agent y, indicating 
    \begin{subequations}
    \begin{align}
        &-f_{\mathrm{x},1}(x_{h+2}) = -f_{\mathrm{x},1}(x_{h+1}) < -f_{\mathrm{x},1}(x_{h}),\label{equaltya}\\
        &-f_{\mathrm{y},2}(y_{h+2}) < -f_{\mathrm{y},2}(y_{h+1}) = -f_{\mathrm{y},2}(y_{h}).\label{equaltyb}
    \end{align}  
    \end{subequations}

    Considering trajectory starting from period 2, switching to period 1, then back to period 2 can show the $-f_{\mathrm{x},2}(x_{h+2}) > -f_{\mathrm{x},2}(x_{h})$ and $-f_{\mathrm{y},1}(y_{h+2}) > -f_{\mathrm{y},1}(y_{h})$ following the similar analysis, we omit the redundant math.
    
    Now, let's analyze if the trajectory starts from period 1 and stays more steps in period 2 before going back to period 1. Given the switching pair $\underline{h},\overline{h}$ as described in Theorem \ref{convergency_theorem}, i.e, $\underline{h}=\overline{h}=1, \underline{h}<h<\overline{h},h=2$. 
    For agent y, staying in period 2 gradually increases $y$ until it goes back to period 1 or reaches the critical point of $-f_{\mathrm{y},2}$ in period 2, which means converging to the critical point. Similar to (\ref{equaltyb}), we have 
    \begin{subequations}\label{x_-i_change}
    \begin{align}
        &y_{\overline{h}} > y_{\overline{h}-1} > ...> y_{\underline{h}+1} = y_{\underline{h}},\\
        &-f_{\mathrm{y},2}(y_{\overline{h}}) < 
        -f_{\mathrm{y},2}(y_{\overline{h}-1}) <...
        <-f_{\mathrm{y},2}(y_{\underline{h}+1}) = 
        -f_{\mathrm{y},2}(y_{\underline{h}}). 
    \end{align}
    \end{subequations}
    where the first inequality and last equality equality comes from (\ref{equaltyb}).
    
    For agent x, although $x_{h}$ will increase in period 2, the $x_{h},y_{h}$ still within period 2, and we have 
    \begin{subequations}
        \begin{align}
        \max\{x_{h}&+y_{h}|h \in (\underline{h},\overline{h})\} < \min\{x_{\overline{h}} + y_{\overline{h}}, x_{\underline{h}}+y_{\underline{h}} \} \label{important_h_bar_lower_bar}
    \end{align} 
    Because (\ref{agenti1}) can't be true and (\ref{agent-i1}) is easily satisfied and $y_{h}$ gradually increase in period 2 as described in (\ref{x_-i_change}), the system switch must be activated by agent y. This means
    \begin{align}
        y_{\underline{h}} < \max\{y_{h}|h\in (\underline{h},\overline{h})\} < y_{\overline{h}},\label{agent_-i_h_bar}
    \end{align}
    Combined with (\ref{equaltya}), we have
    \begin{align}
        x_{\underline{h}} > \max\{x_{h}|h\in (\underline{h},\overline{h})\} = x_{\overline{h}},\label{agenti_important}
    \end{align}
    where the first inequality is obtained due to two cases: i) if $x_{\overline{h}} + y_{\overline{h}} < x_{\underline{h}} + y_{\underline{h}}$, and we know $y_{\underline{h}} < y_{\overline{h}}$ from (\ref{agent_-i_h_bar}), thus, $x_{\underline{h}} >x_{\overline{h}}$; ii) if $x_{\overline{h}} + y_{\overline{h}} \geq x_{\underline{h}}+y_{\underline{h}}$, we know (\ref{important_h_bar_lower_bar}) is equivalent to
    \begin{align}
       \max\{x_{h}&+y_{h}|h \in (\underline{h},\overline{h})\} < x_{\underline{h}}+y_{\underline{h}},
    \end{align}
    and due to (\ref{agent_-i_h_bar}), $y_{\underline{h}} <  \max\{y_{h}|h\in (\underline{h},\overline{h})\} $, thus, $x_{\underline{h}} >  \max\{x_{h}|h\in (\underline{h},\overline{h})\}$.
    Thus, according to (\ref{agenti_important}), we have
    \begin{align}
         -f_{\mathrm{x},1}(x_{\overline{h}}) < -f_{\mathrm{x},1}(x_{\underline{h}})
    \end{align}
     \end{subequations}

    This proves the Theorem by showing we are able to choose learning rate $\tau_\mathrm{x},\tau_\mathrm{y}$ that satisfy $-f_{\mathrm{x},1},-f_{\mathrm{y},2}$ reduce and $-f_{\mathrm{x},2},-f_{\mathrm{y},1}$ increase when $X_{1}+x > X_{2}-x$. Reversely, when $X_{1}+x \leq X_{2}-x$, we can also choose learning rate such that $-f_{\mathrm{x},1},-f_{\mathrm{y},2}$ increase and $-f_{\mathrm{x},2},-f_{\mathrm{y},1}$ reduce.
        
\end{proof}

\section{Proof of Theorem \ref{peak_shacing}}
\begin{proof}
We first show the optimal solution from the centralized CP shaving model as described in (\ref{cost_cen}). Under the condition of $r_{\mathrm{x}} + r_{\mathrm{y}} \geq b$, we can easily get both agents' solutions are their critical point, i.e., $x_{\mathrm{cen}}^* = r_{\mathrm{x}},y_{\mathrm{cen}}^* = r_{\mathrm{y}}$ by first-order optimality condition.
\begin{subequations}
Otherwise, the system demand will be balanced in the two periods, i.e., $x + y =b$, and we add the constraints with Lagrange multipliers $\lambda$
\begin{align}
    \mathcal{L}(x,y,\lambda) &= \pi\max\{S_{1}(x,y),S_2(x,y)\}+ \alpha_{\mathrm{x}} x^2 + \alpha_{\mathrm{y}} y^2 + \lambda(b -x - y),\\
    \frac{\partial\mathcal{L}}{\partial x} &= \pm \pi +2\alpha_\mathrm{x} x - \lambda = 0,\\
    \frac{\partial\mathcal{L}}{\partial y} &= \pm \pi +2\alpha_{\mathrm{y}} y - \lambda = 0,\\
    \frac{\partial\mathcal{L}}{\partial \lambda} &= x + y =b,\\
    x &= \frac{\alpha_{\mathrm{y}}}{\alpha_\mathrm{x}+\alpha_{\mathrm{y}}}b,
    y = \frac{\alpha_\mathrm{x}}{\alpha_\mathrm{x}+\alpha_{\mathrm{y}}}b,
\end{align}
\end{subequations}
where the sign of $\pm$ is determined by the CP period $S_1,S_2$. Overall, the solution for the centralized model (\ref{cost_cen}) is
\begin{subequations}\label{centralized_solution}
    \begin{align}  
        x_{\mathrm{cen}}^* &= \frac{\alpha_{\mathrm{y}} b}{\alpha_{\mathrm{x}}+\alpha_{\mathrm{y}}},
        y_{\mathrm{cen}}^*= \frac{\alpha_{\mathrm{x}} b}{\alpha_{\mathrm{x}}+\alpha_{\mathrm{y}}}, 
        &\text{if }0 \leq b \leq r_{\mathrm{x}}+r_{\mathrm{y}} \\
        x_{\mathrm{cen}}^* &= r_{\mathrm{x}}, y_{\mathrm{cen}}^* = r_{\mathrm{y}}, 
        &\text{if }r_{\mathrm{x}} +r_{\mathrm{y}} < b
    \end{align}    
\end{subequations}

Note that the peak shaving effectiveness defined in the theorem statement is equal to directly comparing the $x^* + y^*$ with $x_{\mathrm{cen}}^* + y_{\mathrm{cen}}^*$ because the baseline demand is the same in both models. Combine with the game solution in Theorem \ref{n_agent_game}, under the concave game conditions, we have $x^* = x_{\mathrm{cen}}^*=r_{\mathrm{x}}, y^* = y_{\mathrm{cen}}^*=r_{\mathrm{y}}$, indicating the peak shaving effectiveness at equilibrium equal to 1. Under the non-concave game and quasiconcave game conditions, it is easy to see both game model and centralized model will balance system demand, i.e., $x^* + y^* = x_{\mathrm{cen}}^* + y_{\mathrm{cen}}^* = b$, meaning that the peak shaving effectiveness at equilibrium also equal to 1. This proves the theorem.
\end{proof}

\section{Proof of Theorem \ref{PoA_agent_equity}}
\begin{proof}
Recall the centralized model solution from Theorem \ref{peak_shacing}, the $x_{\mathrm{cen}}^*,y_{\mathrm{cen}}^*$ under quasiconcave and non-concave game conditions is
\begin{align}
    x_{\mathrm{cen}}^* = \frac{\alpha_{\mathrm{y}}b}{\alpha_{\mathrm{x}}+\alpha_{\mathrm{y}}}, 
        y_{\mathrm{cen}}^* = \frac{\alpha_{\mathrm{x}}b}{\alpha_{\mathrm{x}}+\alpha_{\mathrm{y}}}.
\end{align}
Although the game solution is different under these two conditions, we can denote it as $x^*$, and according to the definition of efficiency loss (\ref{poa}), we have 
    \begin{subequations}
        \begin{align}
        P = \frac{f_\mathrm{x}(x^*,y^*) + f_{\mathrm{y}}(x^*,y^*)}{f_{\mathrm{x}}(x_{\mathrm{cen}}^*, y_{\mathrm{cen}}^*) + f_{\mathrm{y}}(x_{\mathrm{cen}}^*,y_{\mathrm{cen}}^*)} 
        = \frac{\pi S + \alpha_{\mathrm{x}} x^{*2} + \alpha_{\mathrm{y}} y^{*2}}{\pi S
            + \alpha_{\mathrm{x}} (\frac{\alpha_{\mathrm{y}}b}{\alpha_{\mathrm{x}}+\alpha_{\mathrm{y}}})^2 
            + \alpha_{\mathrm{y}} (\frac{\alpha_{\mathrm{x}}b}{\alpha_{\mathrm{x}}+\alpha_{\mathrm{y}}})^2}
    \end{align}
        By replacing $b = x^*+y^*$, we have
        \begin{align}
        P &= \frac{\pi S + \alpha_{\mathrm{x}}x^{*2} + \alpha_{\mathrm{y}} y^{*2}}
        { \pi S
        + \frac{\alpha_{\mathrm{x}}\alpha_{\mathrm{y}}^2
        (x^*+y^*)^2 + \alpha_{\mathrm{x}}^2\alpha_{\mathrm{y}} (x^*+y^*)^2 }
        {(\alpha_{\mathrm{x}}+\alpha_{\mathrm{y}})^2}} = \frac{(\pi S + \alpha_{\mathrm{x}}x^{*2} + \alpha_{\mathrm{y}} y^{*2}) (\alpha_{\mathrm{x}}+\alpha_{\mathrm{y}})}
        {\pi S(\alpha_{\mathrm{x}}+\alpha_{\mathrm{y}}) +\alpha_{\mathrm{x}} \alpha_{\mathrm{y}} (x^*+y^*)^2} \nonumber\\
        &= \frac{\pi S(\alpha_{\mathrm{x}}+\alpha_{\mathrm{y}}) + \alpha_{\mathrm{x}} \alpha_{\mathrm{y}} (x^{*2} + y^{*2}) + (\alpha_{\mathrm{x}} x^*)^2 + (\alpha_{\mathrm{y}} y^*)^2}
        {\pi S(\alpha_{\mathrm{x}}+\alpha_{\mathrm{y}})  + \alpha_{\mathrm{x}} \alpha_{\mathrm{y}} (x^{*2} + y^{*2})+
        2\alpha_{\mathrm{x}} \alpha_{\mathrm{y}} x^* y^* },
       \label{poa_quasi}
        \end{align}
        \end{subequations}  
        
        The first two terms in the denominator and nominator are the same, thus, the difference between denominator and nominator is
        \begin{align}
        (\alpha_{\mathrm{x}}x^*)^2 + (\alpha_{\mathrm{y}}y^*)^2 - 2\alpha_{\mathrm{x}}\alpha_{\mathrm{y}}x^*y^* = (\alpha_{\mathrm{x}}x^*-\alpha_{\mathrm{y}}y^*)^2,\label{poa_quasi_last}
        \end{align}     
        We then know $P$ will increase with $(\alpha_{\mathrm{x}} x^* - \alpha_{\mathrm{y}} y^*)^2$, and $P$ is continuous/differentiable function regarding $\alpha_\mathrm{x},\alpha_\mathrm{y}$ and $x^*,y^*$, thus $\partial P / \partial[(\alpha_{\mathrm{x}} x^* - \alpha_{\mathrm{y}} y^*)^2]> 0$. Note that the shifting cost is $\alpha_{\mathrm{x}}x^{*2}$ for agent x and $\alpha_{\mathrm{x}}x^* = \partial (\alpha_{\mathrm{x}}x^{*2})/\partial x^*$ is the marginal shifting cost. 
\end{proof}

\section{Proof of Theorem \ref{PoA_CPgame_type}}
\begin{proof}
Under the concave game condition, from Theorems \ref{n_agent_game} and \ref{peak_shacing}, we have 
        \begin{align}
            x^* &= r_{\mathrm{x}}, y^* = r_{\mathrm{y}},
            x_{\mathrm{cen}}^* = r_{\mathrm{x}}, y_{\mathrm{cen}}^* = r_{\mathrm{y}},
        \end{align}
    which shows $x^* = x_{\mathrm{cen}}^*, y^* =y_{\mathrm{cen}}^*$ and $P =1$, meaning that the concave game condition is equivalent to centralized model.
    
    From Theorem \ref{PoA_agent_equity}, we know the efficiency loss under quasiconcave and non-concave game conditions can be written as (\ref{poa_quasi}), and nominator minus denominator is $(\alpha_{\mathrm{x}} x^* - \alpha_{\mathrm{y}} y^*)^2\geq 0$. Thus, $P \geq 1$ under these two conditions, indicating quasiconcave and non-concave games always cause higher (or equal) anarchy than concave games.

    Then Given fixed $\pi, S,\alpha_\mathrm{x},\alpha_{\mathrm{y}}>0$, the efficiency loss (\ref{poa_quasi}) is only affected by the solution structure $x^*,y^*$. Note that $X_{1},X_{2},Y_{1},Y_{2}$ could be variant such that $X_{1}+X_{2}+Y_{1}+Y_{2} = 2S$, which may cause quasiconcave or non-concave game condition, i.e., whether both agents satisfy the following conditions
    \begin{align}
        -\frac{\pi}{2\alpha_{\mathrm{x}}} \leq b_{\mathrm{x}} = \frac{X_{2}-X_{1}}{2} \leq \frac{\pi}{2\alpha_{\mathrm{x}}}.\label{condition_quasiconcave_game_poa_proof}
    \end{align}

    We basically fixed the other parameters that appeared in the efficiency loss expression of (\ref{poa_quasi}) to explicitly show the influence of the game type change. According to Theorem \ref{n_agent_game}, we know both quasiconcave and non-concave games balance the system demand over the two periods, i.e., $x^* + y^* = b$. For the quasiconcave game, $x^* = b_{\mathrm{x}},y^* = b_{\mathrm{y}}$, and satisfy the (\ref{condition_quasiconcave_game_poa_proof}), while for non-concave game, $x^* = r_{\mathrm{x}}, y^* = b - r_{\mathrm{x}}$, and $b_{\mathrm{x}} > r_{\mathrm{x}}$. Suppose $x^* = r_{\mathrm{x}}, y^* = b - r_{\mathrm{x}}$, then the condition is $b_{\mathrm{x}} >r_{\mathrm{x}}$ and $b_{\mathrm{y}} < b-r_{\mathrm{x}}$, indicating 
    \begin{align}
        (\alpha_{\mathrm{x}} b_{\mathrm{x}} - \alpha_{\mathrm{y}}b_{\mathrm{y}})^2 > (\alpha_{\mathrm{x}}r_{\mathrm{x}} - \alpha_{\mathrm{y}}(b-r_{\mathrm{x}}))^2.
    \end{align} 
    If $x^* = b\pm r_{\mathrm{y}}, y^* = \mp  r_{\mathrm{y}}$, the situation is similar and we omit the redundant math here. Thus, the quasiconcave game causes higher efficiency loss than the non-concave game under fixed $\pi, S,\alpha_\mathrm{x},\alpha_{\mathrm{y}}>0$.
     These finish the proof of the Theorem.
\end{proof}

\section{Proof of Proposition \ref{uniques_exist_theorem}}
\emph{Overview of the proof}: We derive two lemmas to show the existence and uniqueness of NE in the multi-agent CP shaving game $G'$, respectively. We first show the CP shaving game $G'$ exists NE (Lemma \ref{existence}) and show the NE exists in Lemma \ref{existence} is unique (Lemma \ref{uniqueness}). We first introduce two concepts.
    
(1) \emph{Payoff security (Reny~\citep{reny1999}).} 
    Agent $i$ can secure the payoff $f_i(x_i,x_{-i})-\epsilon \in \mathbb{R}$ at $x_i,x_{-i}\in \mathcal{X}$ iff for every $\epsilon>0$, there exist a $\hat{x}_i\in \mathcal{X}_{i}$ such that $f_i(\hat{x}_i,x_{-i}')\geq f_i(x_i,x_{-i})-\epsilon$ for every $x_{-i}'$ in some neighborhood of $x_{-i}$. Furthermore, we say that a game $G'$ is payoff secure iff every agent $i$ can secure payoff for every $x_i \in \mathcal{X}_i$.
 
(2) \emph{Diagonally strictly concave (Rosen~\citep{rosen1965}).} 
        Define the pseudo-gradient of the sum of all agents' payoff functions $\sum_{i\in N}f_i(x_i,x_{-i})$ and total differential operator $\nabla$ as 
        \begin{subequations}
        \begin{align}
            F(x_i,x_{-i}) =[\nabla_1 f_1(x_i,x_{-i}),\cdots, \nabla_{N} f_{N}(x_i,x_{-i})]^T.\label{pseudo_gradient}
        \end{align}
        Then the function $\sum_{i\in N}f_i(x_i,x_{-i})$ is diagonally strictly concave for $x_i,x_{-i}\in \mathcal{X}$ if for every $x_{\mathrm{a}},x_{\mathrm{b}}\in \mathbb{R}^N$ and $x_{\mathrm{a}},x_{\mathrm{b}}\in \mathcal{X}$, we have
        \begin{align}
            (x_{\mathrm{a}}-x_{\mathrm{b}}) F(x_{\mathrm{b}}) + (x_{\mathrm{b}}-x_{\mathrm{a}}) F(x_{\mathrm{a}}) > 0.
        \end{align}
        \end{subequations}  
        
Among them, payoff security means every agents can secure a payoff value in any strategy profile if they have a strategy that provides at least this value, even if other agents slightly change their strategies. We also have the sufficient conditions for a diagonally strictly concave function from Rosen~\citep{rosen1965}, namely that the symmetric matrix $G(x_i,x_{-i})+G^T(x_i,x_{-i})$ be negative definite for $x_i,x_{-i}\in \mathcal{X}$, where $G(x_i,x_{-i})$ is the Jacobian of $F(x_i,x_{-i})$ with respect to all $x_i$. 
We then introduce the existence lemma.
\begin{lemma}\emph{Existence.} \label{existence}
    The multi-agent two-period quasiconcave CP shaving game $G'$ as described in (\ref{multi_agent_game}) and (\ref{quasiconcaveG'}), has a pure-strategy NE $(x_i^*,x_{-i}^*)$ as defined in (\ref{NE_multi_agent}).
\end{lemma}

\begin{proof}
    According to Reny~\citep{reny1999}, a compact, convex, bounded, and quasiconcave game has a pure strategy NE if the sum of the agent's payoff functions is u.s.c. as defined in Definition \ref{Nash_Equ_defi} across the entire strategy set and if the game satisfies payoff security as defined above. Intuitively, the game is compact, convex, and bounded, meaning that the strategy spaces $\mathcal{X}i, i\in N$ are compact and convex, while the payoff functions $f{i}, i\in N$ remain bounded due to the quadratic shifting penalty term.
    Thus, to establish the existence of a pure-strategy NE, we focus on verifying two key conditions: (a) the sum of the agents' payoff functions, $\sum_{i\in N}f_i(x_i,x_{-i})$, is u.s.c. for all $x_i, x_{-i} \in \mathcal{X}$; and (b) the game $G'$ satisfies the payoff security condition.
    
    (a). By summing all agent's payoff functions, 
        \begin{align}
            \sum_{i\in N}f_i(x_i,x_{-i}) &= - \pi\sum_{i\in N}(X_{i,1} +x_i) I(x_i,x_{-i}) 
            -\pi \sum_{i\in N}(X_{i,2} -x_i) (1-I(x_i,x_{-i})) 
            - \sum_{i\in N}\alpha_{i}x_i^2   \nonumber \\
             &= -\pi S_1 I(x_i,x_{-i}) -\pi S_2 (1-I(x_i,x_{-i})) - \sum_{i\in N}\alpha_{i}x_i^2.  \label{sum_payoff}
        \end{align}
    This function is u.s.c. as per Definition \ref{Nash_Equ_defi} because, on the switching surface where $S_1 = S_2$, the cost term satisfies $\pi S_1 = \pi S_2$.

    (b). We now demonstrate that each agent is payoff secure by analyzing their strategic choices. Given any strategy pair $(x_i,x_{-i})$ and for any $\epsilon>0$, agents $i,-i$ can secure payoffs of at least $f_i(x_i,x_{-i})-\epsilon$ and $f_{-i}(x_i,x_{-i})-\epsilon$, respectively. Note that here $-i$ denotes all agents but $i$.
    
    Suppose the CP period is 1, meaning $S_1 \geq S_2$. If any agent among $-i$ slightly increases their demand shift $x_{-i}$, the CP period remains 1, and agent $i$'s payoff remains unchanged if they maintain their strategy $x_i$.  Conversely, if any agent among $-i$ slightly reduces $x_{-i}$, the CP period may switch to 2. However, agent $i$ can still secure a payoff of at least $f_i(x_i,x_{-i})-\epsilon$ by making a slight reduction in their strategy $x_i$. Specifically, there exists a sufficiently small $\delta>0$ such that $f_i(x_i+\delta,x_{-i}') > f_i(x_i,x_{-i})-\epsilon$ for all cases where the updated demand condition $S_1'=X_{i,1}+x_i+\delta+\sum_{-i\in N}(X_{-i,1}+x_{-i}')>S_2$ still holds. This confirms that agent $i$ is payoff secure. A similar argument applies symmetrically to all other agents $-i$ and the same reasoning extends to the case where the CP period is 2.
    Thus, the CP shaving game is payoff secure. 

    By Reny's theorem~\citep{reny1999}, we conclude that the CP shaving game has a pure-strategy NE.
\end{proof}

\begin{lemma}\emph{Uniqueness.} \label{uniqueness}
    The pure-strategy NE $(x_i^*,x_{-i}^*)$ as described in Lemma \ref{existence} is unique and obtained when (\ref{unique_NE_condition}) hold.
\end{lemma}

\begin{proof}
    According to Rosen's method~\citep{rosen1965}, if the sum of the agents' payoff functions is diagonally strictly concave, as defined earlier, then the equilibrium point established in Lemma \ref{existence} is unique, provided that the constraints in the strategy set are concave. That is, if the strategy space is defined as $\mathcal{X}=\{x|h(x) \geq 0 \}$, where $h(x)$ is a concave function, the uniqueness of the equilibrium follows.

    We first study the uniqueness of pure-strategy NE in each period, with the strategy set $\mathcal{X}_{\mathrm{cp}1},\mathcal{X}_{\mathrm{cp}2}$. In strategy set $\mathcal{X}_{\mathrm{cp}1}$, we have 
        \begin{align}
            \sum_{i\in N}f_i(x_i,x_{-i})
            = -\pi\sum_{i\in N}(X_{i,1}+x_i) - \sum_{i\in N}\alpha_{i}x_i^2  .
        \end{align}
    The Hessian $H$ with respect to $x_i$ is $G = G^T = -2\text{diag}(\alpha_1,\cdots,\alpha_N)$,    
    where $\text{diag}(\cdot):\mathbb{R}^N \to \mathbb{R}^{N \cdot N}$.
    This shows that $G+G^T$ is clearly negative definite, confirming that the sum of the agents' payoff functions is diagonally strictly concave. By Rosen's method, this ensures the existence of a unique NE within the strategy set $\mathcal{X}_{\mathrm{cp}1}$. Applying the same reasoning, we can similarly conclude that a unique NE exists within the strategy set $\mathcal{X}_{\mathrm{cp}2}$. 
    
    The primary challenge in proving this lemma lies in analyzing the transition between the two strategy sets. Specifically, we need to show the existence of strategies $x_i',x_{-i}',x_i'',x_{-i}''\in \mathcal{X}$ such that $S_1( x_i',x_{-i}')-S_2( x_i',x_{-i}')\geq 0, S_1(x_i'',x_{-i}'')-S_2(x_i'',x_{-i}'')<0$.

             
    Then, we leverage the two-agent solution from Lemma \ref{lemma_NE2agent} to extend our analysis to the $N$-agent setting. Given that all agents are capable, we define a \emph{partition} of the set $N$ into two disjoint subsets, denoted as $N_{\mathrm{a}}, N_{\mathrm{b}}$, such that $N_{\mathrm{a}}\cup N_{\mathrm{b}} = N, N_{\mathrm{a}}\cap N_{\mathrm{b}} = \emptyset$. This partitioning ensures that the total demand of these two subsets is unequal in at least one period, i.e., 
    \begin{subequations}
        \begin{align}
            \{\sum_{i\in N_{\mathrm{a}}}(X_{i,1}+x_i)\neq \sum_{i\in N_{\mathrm{b}}}(X_{i,1}+x_i)\}
            \cup\{\sum_{i\in N_{\mathrm{a}}}(X_{i,2}-x_i)\neq \sum_{i\in N_{\mathrm{b}}}(X_{i,2}-x_i)\}.
        \end{align}
        Then, we treat these two sets as two aggregated agents with strategies $x_{\mathrm{a}},x_\mathrm{b}$ and apply the two-agent solutions from Lemma \ref{lemma_NE2agent} again. By doing so, we establish that these two sets will balance their total demand over the two periods, i.e., the following is true for sets $N_{\mathrm{a}}, N_{\mathrm{b}}$,
        \begin{align}
            x_{\mathrm{a}}^* = \frac{\sum_{i\in N_{\mathrm{a}}}X_{i,2}-\sum_{i\in N_{\mathrm{a}}}X_{i,1}}{2} = \frac{\sum_{i\in N_{\mathrm{a}}}(X_{i,2}-X_{i,1})}{2} = 
            \sum_{i\in N_{\mathrm{a}}}b_{i} = \sum_{i\in N_{\mathrm{a}}} x_i^*.
        \end{align}          
        Then, applying the same partitioning process to $N_{\mathrm{a}}, N_{\mathrm{b}}$, we obtain the subsets $N_{\mathrm{a,a}}, N_{\mathrm{a,b}}$ and $N_{\mathrm{b,a}}, N_{\mathrm{b,b}}$. Applying the two-agent solutions again, we establish that these subsets will also balance their demand over the two periods, i.e., $\sum_{i\in N_{\mathrm{a,a}}} x_i^* = \sum_{i\in N_{\mathrm{a,a}}}b_{i}$ for each subset. Iteratively applying this partitioning process to each subset, we continue refining them into smaller disjoint sets. Eventually, as the subsets reduce to singletons, we obtain the final equilibrium solution for each individual agent: $x_i^* = b_{i},i\in N$ and 
        \begin{align}
            \sum_{i\in N}(X_{i,2} - x_i^*) = S_2 = \sum_{i\in N}(X_{i,1} +x_i^*)= S_1.
        \end{align}
    \end{subequations}
             
    This means during the game, periods 1 and 2 are active interactively, and the agent's strategy set switches between $\mathcal{X}_{\mathrm{cp}1}$ and $\mathcal{X}_{\mathrm{cp}2}$ and finally converges at the connecting point of both periods, i.e., $S_1=S_2$. This means there is a unique equilibrium point in the entire strategy set $\mathcal{X}$. 

    Thus, we conclude that the NE as described in (\ref{unique_NE_condition}) is unique and proves this Lemma.
\end{proof}
   
\begin{proof}[Proof of Proposition \ref{uniques_exist_theorem}] We first show the quasiconcave multi-agent CP shaving game $G'$ as described in (\ref{multi_agent_game}) and (\ref{quasiconcaveG'}) has pure-strategy NE based on Lemma \ref{existence}. Then, we prove the NE existed in Lemma \ref{existence} is unique based on Lemma \ref{uniqueness} and the unique NE is obtained as described in (\ref{unique_NE_condition}).
\end{proof}

\section{Proof of Proposition \ref{N_agent_extropolate}}
\begin{proof}
    We first assume two virtual agents as CP agent and non-CP agent. We denote their strategy as $x_{\mathrm{cp}}$, baseline demand as $X_{\mathrm{cp},2} = \sum_{i\in N_\mathrm{cp}}X_{i,2},X_{\mathrm{cp},1} = \sum_{i\in N_\mathrm{cp}}X_{i,1}$, balance point as $b_{\mathrm{cp}} = (X_{\mathrm{cp},2} - X_{\mathrm{cp},1})/2$, and critical point as $r_{\mathrm{cp}} = \sum_{i\in N_\mathrm{cp}}r_{i}$ for CP agent, and change the subscript to $\mathrm{ncp}$ for non-CP period agent. Note that for CP agent $X_{\mathrm{cp},1}<X_{\mathrm{cp},2}$ (the same definition with individual CP-period agent) and non-CP agent vice verse, we also have $b_{\mathrm{cp}}+b_{\mathrm{ncp}} = b$. 
    
    Then, from Theorem \ref{n_agent_game}, we know that these two virtual agents will balance the system demand at equilibrium. The equilibrium outcome depends on which agent reaches its critical point first. If the CP agent reaches its critical point first, it indicates that the non-CP agent's demand shifting cannot be fully offset by the CP agent. In other words, even if the CP agent shifts the maximum possible amount of demand, it is still unable to completely flatten the system demand across the two periods. In this case, the optimal strategy is for the CP agent to shift demand up to its critical point, while the non-CP agent balances the remaining unbalanced system demand. We can express the condition as 
    \begin{subequations}
    \begin{align}
        r_{\mathrm{cp}} < b_{\mathrm{cp}}, 0 \leq b \leq r_{\mathrm{ncp}}+r_{\mathrm{cp}},\label{condition_proposition12}
    \end{align}
    and the NE of these two virtual agents are
    \begin{align}
        x_{\mathrm{cp}}^* =\min\{r_{\mathrm{cp}},b_{\mathrm{cp}}\} = r_{\mathrm{cp}}, x_{\mathrm{ncp}}^* = b -r_{\mathrm{cp}}, 
    \end{align}     
    \end{subequations}
    
    We then aggregate the CP-period agents and non-CP-period agents into the CP-period agent set and non-CP-period agent set, respectively, while keeping the classification of each agent fixed based on their baseline demand conditions. The baseline demand, balance point, and critical point of each agent set are equivalent to those of the corresponding virtual agent. However, they are not entirely identical, as their strategy structures differ. The strategy for the CP-period agent set is determined by the individual decisions of its agents, denoted as $x_i^*,i\in N_{\mathrm{cp}}$, and similarly for the non-CP-period agent set. Despite this distinction, having the same baseline demand and critical point ensures that their best response rationales are aligned. Essentially, just as the virtual CP agent benefits from shifting demand away from the CP period, the agents within the CP-period agent set follow the same strategic behavior.

    By the model formulation, each agent's maximum shifting capacity is $b_{i}$, leading to the best strategy $x_i^* = \min{r_{i},b_{i}}, i\in N_{\mathrm{cp}}$. Furthermore, from condition (\ref{condition_proposition12}), we observe that although the CP-period agent set and the virtual CP agent share the same baseline demand conditions, the CP-period agent set cannot achieve the optimal strategy of the virtual CP agent due to the following
    \begin{align}
       \sum_{i\in N_\mathrm{cp}} \min\{r_{i},b_{i}\} < \sum_{i\in N_\mathrm{cp}}r_{i} = r_{\mathrm{cp}} = x_{\mathrm{cp}}^* < \sum_{i\in N_\mathrm{cp}}b_{i}.\label{new_condition}
    \end{align}
    Since both the virtual CP agent and the CP-period agent set benefit from shifting demand away from the CP period, their best response is to maximize their shifting amount, as their cost reduction monotonically increases with demand shifting. From (\ref{new_condition}), we observe that the total maximum shifting capacity of the CP-period agent set is lower than the optimal shifting amount of the virtual CP agent. This implies that the best strategy for the CP-period agent set is to fully utilize its shifting capacity, leading to an equilibrium shifting amount of $\sum_{i\in N_\mathrm{cp}} \min\{r_{i},b_{i}\}$. Consequently, each agent within the CP-period agent set follows its individual maximum shifting strategy, resulting in $x_i^* = \min\{r_{i},b_{i}\}, i\in N_{\mathrm{cp}}$.
    
    We then verify that, in this case, the non-CP-period agent set will still balance system demand to maximize its profits. First, we confirm that the non-concave game condition remains valid. Since the conditions for the virtual agents and the agent sets are equivalent, i.e.,
    \begin{align}
        0\leq b 
        \leq  \sum_{i\in N_\mathrm{ncp}}r_{i} + \sum_{i\in N_\mathrm{cp}}r_{i} = \sum_{i\in N}r_{i}
    \end{align}
    Thus, under the condition (\ref{new_condition}), where $\sum_{i\in N_{\mathrm{cp}}}x_i^* = \sum_{i\in N_\mathrm{cp}} \min\{r_{i},b_{i}\}$, the non-CP-period agent set can still balance system demand. Next, we show that the best response rationale remains the same between the non-CP-period agent set and the virtual non-CP agent due to their identical baseline demand and critical point conditions. From Theorem \ref{n_agent_game}, we know that the virtual non-CP agent benefits by shifting demand to the CP period in response to the virtual CP agent's strategy. Similarly, the benefit of the non-CP-period agent set increases as it shifts more demand to the CP period. This not only counteracts the demand shifting of the CP-period agent set but also helps to further balance system demand once the CP-period agent set reaches its maximum shifting capacity. Additionally, since $\sum_{i\in N_\mathrm{cp}} \min\{r_{i},b_{i}\} < \sum_{i\in N_\mathrm{cp}}r_{i}$, we have $b -\sum_{i\in N_{\mathrm{cp}}}\min\{r_{i},b_{i}\} > b - \sum_{i\in N_\mathrm{cp}}r_{i}$, indicating that the non-CP-period agent set can actually achieve greater benefits than the virtual non-CP agent.

    To conclude, the CP-period agents in the CP-period agent set with strategy $x_i^* = \min\{r_{i},b_{i}\},i\in N_{\mathrm{cp}}$ form the best response to the aggregated response of non-CP-period agent set with strategy $\sum_{i\in N_{\mathrm{ncp}}}x_i^* = b-\sum_{i\in N_{\mathrm{cp}}}\min\{r_{i},b_{i}\}$ under the condition of 
    \begin{align}
        \sum_{i\in N_\mathrm{cp}} \min\{r_{i},b_{i}\}  < \sum_{i\in N_\mathrm{cp}}b_{i}, 
        0\leq b 
        \leq \sum_{i\in N}r_{i};
    \end{align}
    otherwise, the best strategy will be 
    \begin{align}
        x_i^* = \max\{-r_{i},b_{i}\},i\in N_{\mathrm{ncp}}, \sum_{i\in N_{\mathrm{cp}}}x_i^* = b-\sum_{i\in N_{\mathrm{ncp}}}\max\{-r_{i},b_{i}\},
    \end{align}
    corresponding to the condition 
    \begin{align}
        \sum_{i\in N_\mathrm{ncp}} \max\{-r_{i},b_{i}\}  > \sum_{i\in N_\mathrm{ncp}}b_{i},
        0 \leq b 
        \leq \sum_{i\in N}r_{i}.
    \end{align}

   Note that here, the set of agents remains fixed, and we focus on the aggregate (set-level) performance of both the CP-period agent set and the non-CP-period agent set. Specifically, when the total demand of the CP-period agent set is higher in the non-CP period, this set effectively behaves like the non-CP-period agent set. Simultaneously, the total demand of the non-CP-period agent set must be higher in the CP period, causing it to act as the CP-period agent set. This dynamic mirrors the two-agent case, where agents only swap strategies when the relative relationship between their individual peak periods and the system CP period changes.

\end{proof}

\end{document}